\documentclass[12pt]{article}
\usepackage{fullpage}
\usepackage{todonotes}
\usepackage{hyperref}
\usepackage{comment}
\usepackage{amsmath}
\usepackage{amssymb}
\usepackage{xspace}
\usepackage{xcolor}
\usepackage{authblk}
\usepackage{float}
\excludecomment{ignore}

\usepackage{algorithm}\usepackage[noend]{algpseudocode}
\usepackage{xparse} 

\algnewcommand\algorithmicswitch{\textbf{switch}}
\algnewcommand\algorithmiccase{\textbf{case}}
\algdef{SE}[SWITCH]{Switch}{EndSwitch}[1]{\algorithmicswitch\ #1\ 
}{\algorithmicend\ \algorithmicswitch}
\algdef{SE}[CASE]{Case}{EndCase}[1]{\algorithmiccase\ #1}{\algorithmicend\ \algorithmiccase}%
\algtext*{EndSwitch}%
\algtext*{EndCase}%
 
 \newtheorem{xtheorem}{Theorem}
\newtheorem{xdefinition}[xtheorem]{Definition}
\newtheorem{xobservation}[xtheorem]{Observation}
\newtheorem{xlemma}[xtheorem]{Lemma}
\newtheorem{xproposition}[xtheorem]{Proposition}
\newtheorem{xcorollary}[xtheorem]{Corollary}
{\hspace*{\fill}\raisebox{-1pt}{\boldmath$\Box$}\end{xdefinition}}
\newenvironment{observation}{\begin{xobservation}\rm}{\end{xobservation}}
\newenvironment{theorem}{\begin{xtheorem}\rm}{\end{xtheorem}}
\newenvironment{lemma}{\begin{xlemma}\rm}{\end{xlemma}}

\newenvironment{corollary}{\begin{xcorollary}\rm}{\end{xcorollary}}
\newenvironment{proof}{\begin{trivlist}\item[]{\bf Proof }}%
{\hspace*{\fill}\raisebox{-1pt}{\boldmath$\Box$}\end{trivlist}}

\newcommand{\initalg}{Initialization Algorithm\xspace}

\newcommand{\para}{Parallel Template\xspace}

\newcommand{\miseta}{\ensuremath{\eta_1}}

\newcommand{\good}{reasonable\xspace}
\newcommand{\VC}{\mathit{\tau}\xspace}
\newcommand{\IS}{\mathit{\alpha}\xspace}

\begin{document}
\title{Distributed Graph Algorithms with Predictions\,\footnote{A brief announcement of this work was published in Proceedings of the ACM Symposium on Principles of Distributed Computing (PODC), 2025, pages 322--325.}}
  \author[1]{Joan Boyar}
  \author[2]{Faith Ellen}
  \author[1]{Kim S. Larsen}
 \affil[1]{University of Southern Denmark}
\affil[2]{University of Toronto}
\date{}

\maketitle

\begin{abstract}
We initiate the study of deterministic distributed graph algorithms
with predictions in synchronous message passing systems.  The process
at each node in the graph is given a prediction, which is some extra
information about the problem instance that may be incorrect.  The
processes may use the predictions to help them solve the problem.  The
overall goal is to develop algorithms that both work faster when
predictions are good and do not work much worse than algorithms
without predictions when predictions are bad.  Concepts from the more
general area of algorithms with predictions, such as error measures,
consistency, robustness, and smoothness, are adapted to distributed
graph algorithms with predictions.

We consider algorithms with predictions for distributed graph problems,
where each node is given a prediction for its output.
We present a framework for evaluating distributed graph algorithms with
predictions and methods
for transforming existing algorithms without predictions
to effectively use predictions. 
Our approach is illustrated by developing algorithms with predictions
for the Maximal Independent Set problem.

We also include a discussion of error measures and demonstrate how fine-tuning an error measure towards a particular problem can yield stronger results about the performance of algorithms for that problem.
\end{abstract}

\newpage

\section{Algorithms with Predictions}

The study of algorithms with predictions is a very active research area, originally
largely inspired by Lykouris and Vassilvitskii~\cite{LV18,LV21}.  Like
most work in this area, they considered online algorithms with
predictions.  However, the idea of using predictions, possibly
generated using machine learning, to improve the performance of
algorithms has been applied to other settings, including sequential
algorithms and data structures, dynamic graph algorithms, algorithmic
game theory, and differential privacy~\cite{ALPS}.  Here, we initiate
the study of deterministic distributed graph algorithms
with predictions in synchronous message passing systems.

The quality of an algorithm with predictions for a problem is evaluated
using the same performance measure
as an algorithm without predictions
for that problem.
For online algorithms with predictions, the performance measure is
generally the competitive ratio, comparing the quality of the
solutions obtained by an online algorithm to those of an optimal
offline algorithm.
For synchronous distributed algorithms, the
performance measure we use is the number of rounds until all
processes terminate.

An algorithm with predictions receives information in addition to a problem's input.
This information may come from a machine learning oracle or some other source that
is treated as a black box. It may help an algorithm solve
the instance by compensating for a lack of knowledge; for example, about
future requests in an online setting or future changes in a dynamic setting,
global information in a distributed
setting, or information that would take too long to compute by an
offline or sequential algorithm.
For example, for the Maximal Independent Set problem, each node 
could 
receive a bit indicating 
whether or not it belongs to a set,
which is predicted to be a maximal independent set of the input graph.
Different predictions are given for different problems.
One can also consider different types of predictions for the same problem.

Predictions may have errors.  An \emph{error measure} is a function
that takes a problem instance and predictions
and produces a value that describes the amount of error in the 
predictions.
In  machine learning and related work, this is often called a loss function.
An algorithm with predictions should 
perform very well when the predictions are good and never
perform significantly worse than a good algorithm without predictions.
It is also a goal to have the degradation in performance not grow too fast as the error in the predictions grows.

\subsection{Distributed Graph Problems with Predictions}
\label{introdistrpred}

An instance of a distributed graph problem is a graph, where
each  node of the graph is represented by a different process.
The outputs of the processes must jointly describe a solution to the
problem.

For example,
in the Maximal Independent Set problem, each process must output a bit 
so that the set of nodes 
represented by
the processes that output 1 is
a maximal independent set of the graph.

In general, predictions
can be any kind of information which is given to the processes at the beginning of the algorithm.
Different predictions
can be given to different processes or all processes could be given the 
same predictions.
In this paper, 
each process is given a prediction of its output for the instance.
Jointly, these predictions constitute a predicted solution,
which is not necessarily a correct solution to the problem instance.
For example, in the Maximal Independent Set problem, 
predictions consist of a bit for each process,
where 1 says that its node is predicted to be in the maximal independent set
and  0 says that it is not.
There are two possible types of errors
that can occur in the predictions
for this problem: 
two adjacent nodes could both receive 1, so the predictions do
not describe an independent set,
or a node and all of its neighbors could receive 0, so the 
predicted set is not maximal.

We consider a synchronous message passing model and look at
four standard graph problems, 
Maximal Independent Set, Maximal Matching,
$(\Delta+1)$-Vertex Coloring and $(2\Delta-1)$-Edge Coloring,
where $\Delta$ is the maximum degree of any node in the graph.

An example of where an algorithm with predictions for Maximal Independent Set (or another problem) may be useful is when 
a maximal independent set has been computed on one network, but now a related network 
is being used. It might have the same set of nodes, but a slightly different set of edges or some nodes (and their adjacent
edges) may have been added or removed.

\subsection{Measuring the Quality of Algorithms with Predictions}
\label{intro-error-measures}
In our paper,
an error measure is a function, $\eta$, that maps a problem instance and predictions for that instance to a 
non-negative integer, called the \emph{prediction error}.
When the problem instance and the predictions are understood, we
use $\eta$ to denote the prediction error.
The error measure should intuitively represent the amount of error in the predictions
for that problem instance and the performance of algorithms for that problem should be measured
as a function of the prediction error.
  
Defining an appropriate error measure
for a problem is often a challenge 
when considering problems with predictions.
An error measure should reflect how far from being correct the predictions are,
with a value of zero reflecting no prediction error. 
It is also important to carefully define an error measure so that it is possible
to differentiate between the performance of good and bad algorithms.
If an error measure returns overly large values,
then all algorithms may perform well
as a function of the prediction error,
whereas, if it returns too small values, then all algorithms may
perform poorly.
The error measure is part of the definition of a problem with predictions.

For problems where the correct predictions are unique,
the quality of arbitrary
predictions can be evaluated by comparing
the predictions to the correct predictions.
For example, a graph with distinct edge weights has a unique minimum spanning tree, so the quality of the predictions could be measured by the number of edges incorrectly predicted to belong or not belong to the spanning tree.
If a graph can have duplicate edge weights, it may have multiple minimum spanning trees, and this method of assessing the quality of the predictions cannot be used.

For the distributed graph problems we consider in this paper,
the correct predictions are not unique.
However, these problems have the nice property that certain partial solutions to
a problem instance can be extended to one or more complete solutions.
Moreover, determining a part of the predictions that form
such a partial
solution to an instance can be computed locally, in a small constant number of rounds.
In particular, an algorithm can determine, in a small constant number of rounds,
whether the predictions for a particular problem instance has no error.

Once such a partial solution has been determined,
the subgraph induced by the nodes in the graph
that have not yet output values
may consist of a number of different components.
Extending the partial solution to a complete solution can be done separately
in each component.
Thus, error measures that are functions of the maximum number of nodes, number of edges, or diameter
of these components, rather than the sum of them over all components,
seem more appropriate for these distributed graph problems.
The worst-case round complexities of our distributed graph algorithms with predictions
are at most small constants times the number of 
rounds taken by existing good distributed algorithms
without predictions
for each problem instance, no matter how bad the
predictions are.

Inspired by developments in online algorithms,
we propose a collection of definitions for formalizing the notions
of good performance of algorithms with predictions in distributed computing.
Consider an algorithm, $A$, that solves a 
graph
problem despite possible errors in the predictions
and an algorithm, $R$, without predictions that solves the same problem.
Let $n$ denote the number of vertices in a graph and let $\eta$ denote the prediction error.
Suppose each process has a distinct identifier from the set $\{1,\ldots,d\}$, where $d \in n^{O(1)}$.
\begin{itemize}
\item
 $A$ has \emph{consistency} $c(n)$ if $A$ terminates in $c(n)$ rounds when $\eta = 0$.
$A$ is \emph{consistent} if
 $c(n)$ is at most a constant times the optimal number of rounds needed to check whether a predicted solution is correct.
\item
$A$ is $f(\eta)$-\emph{degrading} if it terminates within $f(\eta) + c(n) + O(1)$ rounds, where $c(n)$ is the consistency of $A$.
\item
$A$ is \emph{robust with respect to $R$}
if it terminates within $O(g(n))$ rounds,
where $g(n)$ is the round complexity of $R$.
\item
$A$ is \emph{robust} if it is robust with respect to $R$, and $R$ is 
an asymptotically optimal algorithm for the problem without predictions.
\item
$A$ is \emph{smooth with respect to $R$} if it is consistent, robust
with respect to $R$, and $f(\eta)$-degrading, where $f$ is a
function that does not grow too quickly.
\item
$A$ is \emph{smooth} if it is consistent, robust, and $f(\eta)$-degrading,
  where $f$ is a function that does not grow too quickly.
\end{itemize}

In the examples considered in this paper, $c(n)$ is constant
and $f(\eta)$ is in $O(\eta)$.
Sometimes, the maximum degree $\Delta$ of the graph and/or the value $d$ is included as a parameter of the functions $c$, $f$, and $g$.

Note that consistency and robustness are standard terms used in the
online algorithms with predictions literature, so we use them here, even
though they have different meanings in the theory of distributed
computing.
In other contexts, degradation is seldom used and, like here, smoothness is only defined informally.

\subsection{Related Work in Distributed Computing}
 
Distributed algorithms with predictions have been considered before in
a very different setting by Gilbert, Newport, Vaidya, and Weaver
in~\cite{GNVW21} for contention resolution.  In their model, the
algorithm is provided with a distribution that predicts the likelihood
of each network size occurring, and they express the performance 
of an algorithm with
respect to the entropy and the statistical divergences between the
predicted and actual distributions.
Mitzenmacher and Dell'Amico~\cite{MD22} consider
predictions for queuing in distributed systems using simulation-based studies.

Starting from our work, Ben-David, Dzulfikar, Ellen and Gilbert~\cite{BDEG25} studied 
Byzantine agreement with predictions in synchronous message passing systems.
In their model, every honest process is given a (possibly different) prediction about which other processes are faulty.
They proved lower bounds on the round complexity and message complexity as a function of the number of 
incorrect prediction bits, the number of processes, and the number of faulty processes. 
They also presented deterministic algorithms with asymptotically optimal round complexity in both an
unauthenticated system and an authenticated system with signatures.

Algorithms with advice~\cite{FIP10, D09, HKK10, EFKR11, BKKKM17, BFKLM17j}
are provided with some additional \emph{correct} information that improves their performance.
This information is chosen by the algorithm designer.
The goal is to understand how much improvement is possible for a given amount of information
or how much information must be given to significantly improve performance.
Fraigniaud, Ilcinkas, and Pelc~\cite{FIP10} proved asymptotically tight bounds on the total number of bits of advice
needed to be given to processes to perform broadcast and wakeup
using $O(n)$ messages in a graph of size $n$.
Together with Gavoille~\cite{FGIP09}, they
proved lower bounds on the number
of nodes in a cycle or a tree that must be given information 
to asymptotically decrease the number of rounds for a deterministic distributed algorithm to compute a 3-coloring.
Fraigniaud, Korman, and Lebhar~\cite{FKL10} designed a deterministic distributed algorithm that computes a minimum spanning tree of
any weighted graph with $n$ nodes in $O(\log n)$ rounds using a constant number of bits of advice per node.
Ellen, Gorain, Miller, and Pelc~\cite{EGMP21} and Ellen and Gilbert~\cite{EG20} show that there are deterministic algorithms
for broadcasting in synchronous radio networks where nodes do not have identifiers, but are given a small constant number of bits of advice.
Very recently,  Balliu, Brandt, Kuhn, Nowicki, Olivetti, Rotenberg and Suomela~\cite{BBKNORS25} gave deterministic distributed algorithms
to find a 3-coloring of any 3-colorable graph and to find  a $\Delta$-coloring of any $\Delta$-colorable graph with maximum degree $\Delta$
using 1 bit of advice per node with round complexity that is a function of only $\Delta$.

A \emph{locally verifiable graph problem} is a problem whose solution,  distributed among the nodes, can be checked in a constant
number of rounds. This means that, if the solution is correct, all nodes must  accept it and, if it is not satisfied,
then at least one node must reject it. 
A {\em locally checkable proof} was defined by G\"o\"os and Suomela~\cite{GoosS16}.
It is advice that can be given to each node of a graph, so that checking whether
the graph satisfies a certain graph property becomes a locally verifiable problem.
They were concerned with the number of bits of advice that are needed for various graph properties.
A \emph{proof labeling scheme}, introduced by Korman, Kutten, and Peleg~\cite{KKP10}, is a special case of a locally checkable proof, where the checking algorithm at each node can only make use of 
the local state of the node, its advice, and the advice of its neighbors.
Balliu, Hirvonen, Melnyk,  Olivetti, Rybicki and Suomela~\cite{BHMORS22} consider partial solutions to 
 locally verifiable graph problems. They are concerned with how large a neighborhood of a node may need to be modified
 so that the resulting partial solution also includes the output value at this node.
 They call this the \emph{mending radius}.

There are also some research areas within the theory of distributed computing that deal with
incorrect information, but address different questions than we do.
Self-stabilizing algorithms~\cite{Dijkstra,Dolev2000,ADDP19} 
start with nodes in arbitrary states and their goal is to ensure that the states
eventually satisfy a certain property, for example, all adjacent nodes have 
different colors.
The designers of these algorithms often strive to minimize
the size of the state at each node. They are not concerned with
improving the round complexity from configurations that are close to being safe.

Kutten and Peleg~\cite{KP99,KP00} consider the solution to a distributed graph problem
which has been changed at some nodes, so that it is no longer correct.
They want to mend the solution in a number of rounds that is a function only
of the number of nodes that have been changed.
In order to do this, they store redundant information in an auxiliary data structure at each node, typically about the
parts of  the solution at its neighbors and use local voting to restore parts of the solution.

\subsection{Outline and Contributions}
In Section~\ref{section-model}, we
present the computational model.
The maximal independent set problem
is formally defined in  Section~\ref{section-graphs}.
There, we also introduce the concept of an extendable partial solution.
Then, we set up our framework,
using the Maximal Independent Set problem as a
running example.
In Section~\ref{initialization}, we introduce initialization and base
algorithms, which are used to provide consistency.
In Section~\ref{error-measures}, we 
discuss criteria for good error measures and give some examples of different error measures for the Maximal Independent Set problem.
As the final part of the set-up, in Section~\ref{sec-css}, we
define measure-uniform algorithms, which are 
used to obtain algorithms with good degradation.

In Section~\ref{template}, we present four templates
that take and sometimes slightly modify
existing algorithms for graph problems
to produce algorithms with predictions. 
Some of the templates require that the algorithms have certain properties, which we exploit to get
better algorithms with predictions.

Other graph problems with predictions are considered in Section~\ref{section-other}.
Then, in Section~\ref{black-white}, we consider the Maximal Independent Set problem,
but using a different error measure, and then restrict it to rooted trees.
Finally, in Section~\ref{section-open}, we list open problems.

\section{Computational Model}
\label{section-model}

A graph $G = (V,E)$ consists of a set of  nodes, $V \subseteq \{1,\ldots,d\}$,
and a set of edges $E \subseteq \{ \{u,v\}\ |\ u \neq v \in V\}$. 
The number of nodes in $V$ is denoted by $n$ and 
%$d \in n^{O(1)}$
$d$ is an upper bound on the largest identifier
of any node in the graph.

Each node of the graph is a nonfaulty process that may communicate with its neighbors,
those nodes with which it shares an edge.
The number of neighbors of a node is called its degree.
Each node is assumed to know 
its identifier and the identifiers of its neighbors, as well as the values $n$ and~$d$.  
Sometimes 
they
are assumed to know additional information about the graph,
for example, $\Delta$, which is 
the maximum degree of any node in the graph.

Each node $i$
has a local  variable $x_i$, which is a prediction of its part of the solution.  
It also
has a local output variable, $y_i$, (or possibly multiple output variables) used to indicate its part of the
solution to the problem at the end of 
the 
computation.  Initially, the outputs are undefined.
Immediately after node $i$
has assigned values to all its output variables, it terminates.
A node that has not terminated is called \emph{active}.
Computation is synchronous: In each \emph{round}, 
each active node can send a possibly different message to each of its neighbors, receive all
messages sent to it that round from all of its neighbors, do some
computation and update its state, optionally 
assign a value to its local output, and terminate if this is the node's
last output. 
An algorithm terminates when every node has terminated.
In the LOCAL model~\cite{L92}, a message can be arbitrarily long, while in the CONGEST
model, each message is restricted to be  $O(\log n)$ bits long~\cite{Peleg2000}.
Although some of our algorithms will work in the CONGEST model, we 
use the LOCAL model, for simplicity.
All algorithms are assumed to be deterministic.

\section{The Maximal Independent Set Problem and Extendable Partial Solutions}
\label{section-graphs}

In this section, we consider the maximal independent set problem without predictions
and introduce the concept of extendable partial solutions.

An \emph{independent} set of a graph is a subset $I$ of its nodes
such that no two nodes in $I$ are neighbors in the graph.
It is \emph{maximal} if it is contained in no other independent set.
In the \emph{Maximal Independent Set (MIS) problem}, 
each node $i$ must output a bit $y_i$ such that 
no neighbor of a node that outputs 1 also outputs 1
and every node that outputs 0 has a neighbor that outputs 1.
At the end of the computation, 
the nodes that have output 1 must be
a maximal independent set of the graph.

The outputs of some nodes 
form a \emph{partial solution} to a graph problem
if they are a solution to the problem on the subgraph induced by the specified
nodes or edges. 
If a partial solution together with \emph{any} solution to the problem on
the remainder of the graph is a solution to the
problem on the entire graph, we say that this partial solution is {\em
  extendable}. In particular, a solution to the entire graph problem is
trivially an extendable partial solution.

For example, if $I$ is any independent set of a graph, $G$, then the partial
solution in which all nodes in $I$ output 1 and all their neighbors
output 0 is an extendable partial solution. This holds because if
$I'$ is any maximal independent set of the  graph
obtained by removing $I$ and its neighbors from $G$, then $I' \cup I$ is 
a maximal independent set of $G$~\cite{L86}. Conversely, suppose a partial
solution has a node $v$ with output~$1$, but it does not have output~$0$ for
some neighbor $u$ of $v$ in the original graph. Then, this partial solution
is not extendable, because there exists a maximal independent set in the
remainder of the graph that contains $u$.

A related concept is a \emph{greedily completable problem}~\cite{BHMORS22}. It
is a problem such that any partial solution can be extended to a complete solution in a greedy manner, 
considering the nodes in an arbitrary order.
Note that it is not necessary to know an extendable partial solution  to extend it to a complete solution.
However, for a greedily completable problem, knowledge of the partial solution may be used to extend it to
a complete solution. 
Thus, a greedily completable problem may not necessarily have extendable partial solutions.

\section{Initialization Algorithms}
\label{initialization}
To ensure consistency in a distributed graph algorithm with predictions,
each node begins by examining the predictions in its local neighborhood
to determine whether they are consistent with a correct solution.
If so, the node outputs its prediction; otherwise,
it might output some values which can be different from
its prediction. If a node has output its part of the solution, it terminates.
Thus, if the predictions form a correct solution, all nodes terminate.
We require that the set of outputs of all nodes forms an
extendable partial solution.
We call this initial part of a distributed graph algorithm with predictions
an \emph{initialization algorithm}.
The nodes that have not terminated
proceed to compute a solution on the remaining graph. 

For example, there is a simple 3 round initialization algorithm
for MIS, which we call the \emph{MIS Base Algorithm}:
In round 1, nodes exchange their predictions with their neighbors. The nodes that have prediction 1 
and whose neighbors 
have prediction 0 form an independent set $I$. In round 2, every node in $I$ notifies its neighbors, outputs 1, and terminates.
In round 3, every node that is a neighbor of a node in $I$ (and hence has
prediction~0) notifies its neighbors, outputs 0, and terminates.
If the predictions are
correct, then this algorithm outputs the
prediction for every node. Hence, 
any algorithm with predictions that begins with the MIS Base Algorithm
 is consistent.

The extendable partial solution computed by the MIS Base Algorithm has the additional property that every node that outputs a value outputs its prediction.
  Korman, Sereni, and Viennot~\cite{KSV13} called such an algorithm a
  \emph{pruning algorithm} and gave the MIS Base Algorithm as an example.
  In our context,
  they are initialization algorithms where
  the outputs at nodes are the same as the predictions at these nodes.
Consider any node $v$ that outputs~$1$ in an extendable
partial solution produced
by a pruning algorithm for the MIS problem. 
As mentioned earlier, all of its neighbors output~$0$ in this solution.
Furthermore, all outputs for $v$ and
its neighbors agree with the prediction. Therefore, the MIS Base Algorithm
outputs~$1$ for node $v$. Hence, among all pruning algorithms, the
MIS Base Algorithm outputs the largest independent set. However, some
pruning algorithms may produce more outputs than the MIS Base Algorithm.
For example, consider a graph consisting of a line with three nodes, which
are all given prediction~$0$. The partial solution produced by the MIS
Base Algorithm is empty, but the algorithm where the middle node outputs~$0$
is a pruning algorithm.

More generally, a \emph{base algorithm} is an initialization algorithm
which is a simple pruning algorithm for
a problem. It is defined as part of the problem definition.
An \emph{error component}
is a component of the subgraph induced by the nodes that 
would be active after running the base algorithm.
In applications where the outputs are associated with the edges of the graph, instead of its nodes, the error components are the components of the subgraph induced by the
edges whose outputs would not be determined after running the base algorithm.

When the predictions for the entire graph are not completely correct, an
initialization algorithm
is allowed to output values other than the prediction at some nodes.
Because the partial solution computed by an
initialization algorithm is extendable, computing a
solution separately on each component will lead to a correct overall solution.

Another initialization algorithm for the MIS problem, which is not a pruning algorithm,
can be obtained by changing the definition of 
the set $I$ in the MIS Base Algorithm
to be the set of all nodes with prediction 1 whose neighbors with prediction 1 (if any) all have smaller identifiers. We call this the
\emph{MIS Initialization Algorithm}.
The  extendable partial solution obtained by the MIS Initialization
Algorithm always contains the one produced
by the MIS Base Algorithm.

We say an initialization algorithm for the problem is \emph{\good}
if it always produces an extendable
partial solution
that contains the extended partial solution produced by the base algorithm 
and its asymptotic round complexity is the same as the base algorithm.
The MIS Initialization Algorithm is an example of a \good initialization algorithm,
which has the same round complexity as the MIS Base Algorithm.
We use that algorithm for initialization
when we create examples of algorithms with predictions from the templates.
We emphasize that it is important to have a base algorithm for each problem,
from which the error components can be defined.
Based on those error components, we can then define one or more error measures.
Given an error measure, we can then compare
the performance of different algorithms with predictions for that problem,
even if they start with different initialization algorithms.  

\section{Error Measures}
\label{error-measures}

Our goal is to be able to compare different distributed algorithms with
predictions for a given problem, analyzing the round complexity
as a function of the prediction error. There are various ways to define error
measures for distributed graph algorithms with predictions.
In this
section, we discuss a number of possibilities.

A seemingly natural error measure could be defined by letting the error measure, $\eta_H$, be the minimum number of predictions that would have to be
changed to obtain a correct solution. 
For example, for the MIS problem,
this would be the minimum Hamming distance between the characteristic
vector of the prediction
 and the characteristic vectors of the maximal independent sets of the graph.
This error measure is global, rather than local:
if there are multiple error components, it 
counts the total number of errors in all of them,
rather than just considering the component with the 
largest number of errors. This does not reflect the fact that nodes
in these different 
error components can work independently to find solutions for their components.
The sum of the sizes of the error components (rather than the maximum of their sizes)
is also a global error measure, so it suffers from the same weakness as $\eta_H$.

Our approach is to define a measure, $\mu$, on graphs
and define an error measure, $\eta$, by taking the maximum 
of $\mu$ over all error components. These error measures depend on both the graph and the predictions.
Note that
the components of the subgraph induced by the active nodes at the end of
different initialization algorithms may be different than the error components.
Using \good initialization algorithms, each of those components will be a subgraph of an error component.

Our first measure, $\mu_1$, is the number of nodes in the 
graph
and our
first error measure, $\eta_1$, is $\max\{ \mu_1(S)\ |\ S$ is an error component\}.
Note that, for the MIS problem, if a node and one or more of its neighbors in the original graph all have prediction 1,
then these nodes will all be in the same error component.
If a node and all of its neighbors in the graph have prediction 0, then that node will be in some error component.
When the predictions are correct, there are no error components and the prediction error is zero.
Although the base algorithm is used to define this error measure,
an algorithm with predictions can begin with another \good initialization algorithm, instead.

Im, Kumar, Qaem, and Purohit~\cite{IKQP21,IKQP21a} recommended that the value of an error measure 
should not increase when errors are removed. 
In our setting, 
if the subgraph induced by the set of nodes that are active at the end of the base algorithm for one set of predictions
is contained in the subgraph induced by the set of nodes that are active at the end of the base algorithm for another set of predictions,
then the value of the error measure for the first set of predictions
should be at most the value of the error measure for the second 
set of predictions.
This follows whenever the measure $\mu$ from which $\eta$ is defined is \emph{monotone}:
If $S$ is a component of an induced subgraph of a connected graph $G$, then  $\mu(S) \leq \mu(G)$.
Note that $\mu_1$ is monotone.

Another measure of a connected graph
is its diameter.
This seems useful, especially in the LOCAL model,
since a solution for any connected graph 
can be computed in a number of rounds bounded by its
diameter (by collecting the adjacency lists of every reachable node
and then locally computing the solution).
Because a tree is an acyclic connected graph,
the components of any induced subgraph of a tree
have diameter no larger than the diameter of the tree.
Thus, this measure is monotone for trees.
(Results for rooted trees are discussed in Section~\ref{mis-tree}.)
However, it is not monotone for general graphs, since
there are connected graphs that have induced subgraphs with components of much larger diameter.
For example, consider $F_k$, the
wheel with $k$ nodes on the rim and an additional node on each
spoke, illustrated in Figure~\ref{fig:F8}.
The diameter of $F_k$
is $4$, since the path from any node to the center node has length
at most~$2$.
The subgraph induced by the $k$ nodes on the rim has diameter $\lfloor k/2 \rfloor$.
For the maximal independent set problem with predictions, this subgraph can be an error component,
for example, when the center node has prediction 1 and the remaining nodes have prediction 0.
When all nodes have prediction 1, which is a worse prediction, the entire graph is an error component,
but it has smaller diameter.
Thus, the maximum of the diameter of every error component 
should not be
used as an error measure for general graphs.

\begin{figure}[!htb]
\centering
\includegraphics{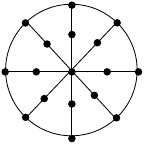}
\caption{The graph $F_8$}
\label{fig:F8}
\end{figure}

In summary, an error measure is required to be the maximum of some
monotone measure of each of the error components. 
Consequently,
if $\eta'$ was computed in the same way as $\eta$, but from the components 
remaining after
running a \good initialization algorithm instead of after running the base
algorithm, then $\eta'\leq \eta$.

For the Maximal Independent Set problem with predictions,
we define another monotone measure and, from it, another error measure.
A \emph{vertex cover} of a graph is a subset of nodes that contains at least one endpoint of
every edge.
It follows that the complement of a vertex cover is an independent set and vice versa.
For any graph, $G$, let $\IS(G)$ be the size of a maximum independent set of $G$
and let  $\VC(G)$ be the size of a minimum vertex cover of $G$.
The closest correct solution in
 an error component, $S$,  in which all nodes have prediction 1
 is a maximum independent set of $S$.
 The nodes with incorrect predictions relative to this solution
 is a minimum vertex cover of $S$. 
 In this case, $\VC(S)$ is the number of incorrect predictions.
The size  of a maximum independent set of an error component
is an upper bound on the size of the maximal independent set
of that component that is found by an algorithm.
Thus, if a node is added to the independent set every constant number of rounds of some algorithm,
then the algorithm performs $O(\IS(S))$ rounds on any error component $S$.
For any graph $G$, let
$$\mu_2 (G) = 2 \min\{ \IS(G), \VC(G)\} \mbox{ and let }  \eta_2 = \max\{ \mu_2(S)\ |\ S \text{ is an error component} \}.$$
Since the complement of the set of nodes in an independent set of a graph is a vertex cover of that graph,
for any component, $S$,   at least one
of $\IS(S)$ and $\VC(S )$ is at most half the number of nodes in $S$.
Hence,
\[\mu_2 \leq \mu_1,\]
and,
for any instance of maximal independent set with predictions,
\[\eta_2 \leq \eta_1.\]
Therefore, showing that an algorithm is $\eta_2$-degrading is at least
as strong as showing that it is $\miseta$-degrading.
Note that, when $G$ is a clique, $\IS(G) = 1$ and, when $G$ is a star, $\VC(G) = 1$, so
$\mu_2$ can be significantly smaller than $\mu_1$ and
$\eta_2$ can be significantly smaller than $\miseta$.

Another possible error measure for the MIS problem is based
on a different definition of an error component.
A {\em black (white) component} is a component of the graph induced by the nodes
that have prediction 1 (0) and were active 
at the end of the base algorithm.
We define the error measure, $\eta_{bw}$,
to be the number of nodes in the largest black or white component
that would remain in the graph
if the base algorithm were run.
Note that, for any instance, the prediction error, $\eta_{bw}$, is
bounded above by $\miseta$.
In some instances, for example
when all nodes receive the same prediction, $\eta_{bw} = \miseta$.
However, there are interesting instances where the value of $\eta_{bw}$ is much smaller than the value of $\miseta$.
This can lead to faster algorithms.
For example, consider a two-dimensional grid with $n$ nodes,
where the nodes with coordinates in 
$\{(i,j)\ |\ i,j  \bmod 4 \in\{ 0,1\}\} \cup  \{(i,j)\ |\ i,j  \bmod 4 \in\{ 2,3\}\}$
are colored black and the remaining nodes are colored white. This is illustrated in Figure~\ref{fig:grid}.
For this instance, $\miseta =n$, but $\eta_{bw} = 4$.
In general, splitting the error
components into smaller ones based on the predictions the nodes receive is a symmetry breaking mechanism.
We give examples 
of MIS algorithms with predictions that use
this error measure in Section~\ref{black-white}.
\begin{figure}[!htb]
\centering
\includegraphics[scale=0.6]{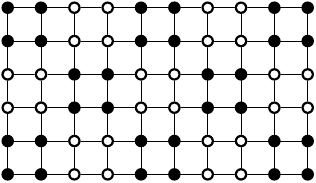}
\caption{The black and white components of an instance of the grid graph}
\label{fig:grid}
\end{figure}

\subsection{Error Measures Used for Graph Problems in Other Settings}
Graph problems with predictions have been considered in order to improve the
running time of classic and dynamic graph algorithms, and the competitive ratio of online graph algorithms.

The running times of algorithms for various graph problems have been
improved using predictions.
Chen, Silwal, Vakilian and Zhang~\cite{CSVZ22} consider problems
where predicting the dual solution is helpful. For different problems, they
consider error measures based on  different  $L^p$-norms of the difference between the
predicted dual solution
and an optimal dual solution. 

Moseley, Niaparast and Singh~\cite{MNS25} consider predictions for the Global Minimum Cut problem that indicate whether each edge is in the minimum cut. They compare the predictions 
to an arbitrary minimum cut, using two error measures together,
one error measure that counts the number of edges that are
incorrectly predicted to be in this minimum cut and another error measure that counts the number
of edges that are incorrectly predicted to not be in this minimum cut.
In both cases, the error measure is normalized by the weight of the minimum cut.
They also consider predictions that give the probability that each edge is in the minimum cut.
The two error measures they use in this setting are natural extensions of their other error measures.

Max-flow has been considered with different error measures, 
where the predictions are
 always an amount of flow for each edge. In the papers 
by Davies, Moseley, Vassilvitskii and Wang~\cite{DMVW23} and Polak and Zub~\cite{PZ24},
the error measure is the minimum $L^1$ norm of the difference between the predicted flow, which may not be feasible, and a closest maximum flow.
In another paper, by Davies, Vassilvitskii and Wang~\cite{DVW24},
the error measure considers both the minimum amount of flow needed to be added to the prediction to saturate some $s-t$ cut and how far the prediction is from being a feasible flow
(by summing the absolute value of the difference between the incoming and outgoing flow at each vertex).

Predictions have also been considered for dynamic graph algorithms, especially incremental graph algorithms.
Incremental Approximate Shortest Path problems with predictions (among other problems) were considered by van den Brand, Forster, Nazari and
Polak~\cite{BFNP24}, Henzinger, Saha, Seybold and Ye~\cite{HSSY24}, and
McCauley, Moseley, Niaparast, Niaparast, and Singh~\cite{MMNNS25}. In all three papers, one prediction considered is the order of arrivals of edges, and the
error measure is the maximum, over all edges, of the absolute value of the difference between its predicted and actual position in the order of arrival.
Incremental problems for directed graphs are considered by McCauley, Moseley, Niaparast, and Singh~\cite{MMNS24}. 
The prediction for each vertex is the number of edges that are ancestors of the vertex and the  number that are descendants of the vertex in the final graph.
The error measure is the maximum over all vertices of the sum of the absolute values of the differences between the predicted and actual values of their numbers of ancestors and  descendants.

Antoniadis, Broersma and Meng~\cite{ABM24} considered online graph coloring algorithms with predictions,
where each vertex that arrives has a predicted color.
The error measure they use
is the number of incorrect predictions as compared to a closest
optimal coloring, i.e., a coloring that uses the fewest number of colors.

Online algorithms on metric spaces have been considered, where
the predictions are the locations at which requests will arrive.
Note that this setting applies to graph problems, letting the distance 
between two vertices in the metric space
be the length of the shortest path between them in the graph.
Azar, Panigrahi, and Touitou~\cite{APT22} define an error measure based on 
the minimum cost of a matching between some predicted locations and some actual locations,
as well as the number of predicted and actual locations that were
not included in the matching.
Bernardini, Lindermayr, Marchetti-Spaccamela, Megow, Stougie and Sweering~\cite{BLMMSS22}
define error measures based on 
covering the unexpected actual requests by predicted requests
and 
covering absent predicted requests by actual requests.

Online graph algorithms with predictions have been considered where items
arrive one a time and must be either accepted for or rejected from the final solution.
Boyar, Favrholdt, Kamali and Larsen~\cite{BFKL23} studied the problem of obtaining a matching
of large size, where the items are edges of the graph.
At the beginning, an algorithm is given a prediction of 
which edges will arrive.
Thus, prediction errors are predicted edges that do not arrive in the actual input or arriving edges that were not predicted. 
Their error measure is the maximum size of a matching of the incorrectly predicted edges,
normalized by the maximum size of a matching of all edges.
Berg, Boyar, Favrholdt and Larsen~\cite{BBFL25} and Berg~\cite{B25} consider the 
Bounded-Degree Vertex Cover and the Bounded-Degree Independent Set problems, respectively,
where the items are the vertices (together with the edges adjacent to previously arrived vertices).
Each arriving vertex comes with a prediction
of whether the vertex is in the given set, so should be accepted, or not in the set, so should be rejected.
Two error measures are used with possibly different weights, the number of incorrect predictions of vertices to accept and the number of incorrect
predictions of vertices to reject.

\section{Measure-Uniform Algorithms}
\label{sec-css}

Korman, Sereni and Viennot~\cite{KSV13} call a graph algorithm \emph{uniform} with respect to a graph parameter
if the round complexity of the algorithm depends on the value of that parameter, but the algorithm
does not require that nodes know even an upper bound on the value of  the parameter.
If such an algorithm is run on a subgraph, its round complexity is a function of the
value of the parameter for the subgraph, 
rather than for the entire graph.

A graph algorithm is \emph{measure-uniform} with respect to 
measure $\mu$ if it is uniform with respect to $\mu$
and its round complexity does not depend on the value of any other graph parameters.
For example, if the round complexity of the algorithm depends only on the number of nodes in the graph,
then the round complexity of the algorithm when performed on (a subgraph of) an error component
will be bounded above by a function of the prediction error $\miseta$, which could be much smaller than the
number of nodes in the original graph. 
However, if the round complexity of the algorithm also depends on a parameter such as the size of the domain from which the 
identifiers are chosen or on the largest
identifier, the value of this parameter
may be the same as in the original graph.

A measure-uniform algorithm can be used as part of a graph algorithm with predictions
to separately compute a solution on each error component independently.
We use measure-uniform algorithms in our framework to obtain algorithms with predictions that have good degradation functions.

An example of a measure-uniform algorithm for the MIS problem
is to repeatedly choose an independent set of nodes in the remaining graph, add these nodes
to the independent set, and eliminate them and their neighbors from the graph~\cite{KW85}.
One way to deterministically choose an independent set is to take every node whose identifier
is larger than the identifiers of all its active neighbors. We call this  the  \emph{Greedy MIS Algorithm}.
Pseudo-code is presented in Algorithm~\ref{Identifier_MIS}.
It assumes that each
node knows the identifiers of all its neighbors.

\begin{algorithm}
  \label{Identifier_MIS}
\caption{Greedy MIS Algorithm, code for node $p_i$}\label{greedyMIS}
\begin{algorithmic}[0]
\For {$r = 1,\ldots$}
\State {\bf round} $2r-1$:
\If {all active neighbors of $p_i$ have identifiers smaller than $i$}
\State it notifies all its active neighbors that it is in the independent set,
\State outputs 1, and terminates
\EndIf
\State {\bf round} $2r$:
\If {$p_i$ received a message from a neighbor during round $2r-1$}
\State it notifies all its active neighbors that it is not in the independent set,
\State outputs 0, and terminates
\Else
\For {each neighbor $p_j$ from which $p_i$ received a message during round $2r$}
\State $p_i$ locally records that $p_j$ has terminated
\EndFor
\EndIf
\EndFor
\end{algorithmic}
\end{algorithm}

At the end of each even numbered round, 
a node has output 0 if and only if it is the neighbor
of a node that has output 1.
Thus, the nodes 
that have output 1 form an independent set of the graph and
the values that have been output by the end of each even numbered round form an extendable partial solution.
At the end of the algorithm, when all 
nodes
have terminated,
the set of nodes 
that have output 1
is a maximal independent set of the graph.

\begin{lemma}
\label{greedyMISuppermu1}
When performed on a graph, $G$,
the Greedy MIS Algorithm has round complexity at most
$\max\{\mu_1(S) \ |\ S$ is a component of $G$\} and
 it is measure-uniform with respect to $\mu_1$.
 \end{lemma}

\begin{proof}
During each odd numbered round,
at least one  node
in each component of the remaining graph has an identifier
that is larger than all its neighbors and, hence, terminates.
If a component  contains more than one node, then at least one 
node in the component receives a message
from a neighbor and, hence, terminates during the subsequent even numbered round.
Thus, the Greedy MIS Algorithm has round complexity at most $s$ when
performed on a component, $S$, with $s = \mu_1(S)$ nodes.
Since the code does not mention any graph parameters, it is uniform with respect to $\mu_1$.
Hence, it is measure-uniform with respect to $\mu_1$.
\end{proof}

\begin{lemma}
\label{greedyMISuppermu2}
When performed on a graph, $G$,
the Greedy MIS Algorithm 
has round complexity at most  $\max\{\mu_2(S) +1\ |\ \textrm{ S is a component of } G\}$ and
 it is measure-uniform with respect to $\mu_2$.
\end{lemma}
\begin{proof}
Let $S$ be a  component of $G$
and  let $S'$ be the subgraph of $S$ induced by
the  active nodes in $S$ immediately
before some odd numbered round of the Greedy MIS Algorithm.
During the next round,
at least one   node $v$ in $S'$
is chosen to be in the maximal independent set and, in the next
  round, all of $v$'s neighbors in  $S'$
 are excluded from the maximal independent set,
creating an extendable partial solution. 
Let $S''$ be the subgraph induced by the remaining active nodes
in~$S$.

If  $S''$  has an independent set, $M$, of size at least $\IS(S')$,
then $M \cup\{v\}$ is an
independent set in $S'$ of size at least $\IS(S')+1$, which
contradicts the definition of $\IS(G')$.
 Thus, $\IS(S'') < \IS(S')$. 
 Since a nonempty graph has an independent set of size at least 1,
 all nodes of $S$ terminate within $2\IS(S)$ rounds of the Greedy MIS
 Algorithm.
 
 Let $V'$ be a vertex cover of $S$ of size $\VC(S)$.
 If $S'$ contains at least one edge, then at least one node, $v$, with a  neighbor, $u$, in $S'$
  is chosen to be in the maximal independent set.
 At least one of $v$ and $u$ must be in $V'$, so
$V(S'') \cap V' $ is a proper subset of $V(S') \cap V'$.
 Thus,
  after $2\VC(G)$ rounds, all nodes in $V'$ have terminated and their
  edges have been removed, so there are no edges remaining.
In a graph with no edges, all nodes terminate within one round of 
the Greedy MIS Algorithm, so
all nodes of $S$ terminate within $2\VC(S)+1$ rounds of the Greedy MIS Algorithm.

Therefore, the Greedy MIS Algorithm has round complexity at most
the maximum over all components $S$ of the graph of
$\min\{2\IS(S), 2\VC(S) +1\} \leq \mu_2(S)+1$.
Since the code does not mention any graph parameters, it is uniform with respect to $\mu_2$.
\end{proof}

The Greedy MIS algorithm is an asymptotically optimal 
measure-uniform algorithm with respect to
the  measures $\mu_1$ and $\mu_2$.
We establish this result
using the following result from Ramsey Theory~\cite{ramsey,ramseya}.

\begin{theorem}
For all positive integers $K$, $L$, and $M$, there is a positive integer $R(K,L,M)$ such
that for every set $S$
of size at least $R(K,L,M)$ and every function $f$ that maps the $K$-element
subsets of $S$ to $\{ 0,1,\ldots,M-1\}$, there is a subset $S' \subseteq S$
of size at least $L$ such that $f$ has the same value on every $K$-element subset of $S'$.
\label{ramseythm}
\end{theorem}

Consider the graph consisting of $n$ nodes connected in a line, where every node initially knows its own identifier and the identifier of  its left neighbor and
the identifier of its right neighbor, if these neighbors exist. We prove lower bounds on the number of rounds to solve the 
MIS problem on this graph, starting with a lower bound 
for coloring its nodes with at most 3 colors so that no two adjacent nodes get the same color.

\begin{lemma}
Every deterministic 
algorithm for coloring the nodes of an $n$ node line with  3 colors 
requires at least $(n-3)/2$ rounds,
if its round complexity does  not depend on the domain 
 from which the node identifiers are chosen or $d$ is a
  sufficiently large function of $n$.
\label{css-vc}
\end{lemma}

\begin{proof}
To obtain a contradiction,
suppose there is a deterministic
algorithm $A$ that always colors the nodes of this graph using  3 colors in $T(n) < (n-3)/2$ rounds.
Then $2T(n) +4  \leq n$.
In addition to its own identifier, each node initially only knows the identifiers of its left and right neighbors.
By the end of the algorithm,
each node only learns the identifiers of the nodes at distance at most  $T(n) + 1$ away.
Nodes that are neither among the first $T(n)+1$ nodes nor the last $T(n)+1$ nodes
(which includes the nodes in positions $T(n) +2$ and $T(n) +3$) cannot determine their position relative to the end of the line.
Each such node only knows its own identifier and the identifiers of the $T(n)+1$ to their left and right,
so it decides what color to output as a function of this sequence of identifiers.
For these nodes, algorithm $A$ can be viewed as a function that maps each sequence of $2T(n)+3$ identifiers to one of the three colors.
When restricted to labelings of the line with identifiers that are strictly increasing from left to right,
it can be viewed as a function that maps each subset of $2T(n)+3$ identifiers to one of the three colors.

By Theorem~\ref{ramseythm}, there is a (large) integer $r=R(2T(n)+3, 2T(n)+4,3)$ and a subset
$S' \subseteq \{ 1,\ldots,r\}$ of identifiers of size at least $2T(n)+4$ such that $A$ maps every
$2T(n)+3$ element subset of $S'$ to the same color.
Consider a line of length $n \geq 2T(n)+4$ whose first $2T(n)+4$ nodes have strictly increasing
identifiers from $S'$.
Then $A$ colors the nodes in positions $T(n)+2$ and $T(n)+3$  with the same color, contradicting the correctness of $A$.
\end{proof}

\begin{lemma}
Every deterministic 
algorithm for computing a maximal independent set of an $n$ node line requires at least $(n-5)/2$ rounds,
if its round complexity does not depend on the domain
from which the node identifiers are chosen
or $d$ is a sufficiently large function of $n$.
\label{css-mis}
\end{lemma}

\begin{proof}
To obtain a contradiction,
suppose there is a deterministic algorithm 
that always computes a maximal independent set of this graph in fewer than $(n-5)/2$ rounds.
Instead of outputting a value and terminating,  each node
broadcasts this value to its neighbors in the next round.
Then the nodes in the independent set output 0. 
Nodes that are not in the independent set output 1 if they have a left neighbor that is in the independent set;
otherwise, they output 2.
This is an algorithm that 3-colors the nodes of an $n$ node line with 3 colors in fewer than  $(n-3)/2$ rounds, contradicting Lemma~\ref{css-vc}.
\end{proof}

\begin{theorem}
  The Greedy MIS Algorithm is an asymptotically optimal measure-uniform
algorithm with respect to $\mu_1$ and $\mu_2$.
\end{theorem}
\begin{proof}
Consider any  algorithm for the MIS problem that is measure-uniform with respect to $\mu_1$ or $\mu_2$.
Then its round complexity 
does not depend on the domain from which the node identifiers are chosen.
By Lemma~\ref{css-mis}, when
run on an $n$ node line, $L_n$, this algorithm takes at least $(n-5)/2 \in \Omega(\mu_1(L_n)) \subseteq \Omega(\mu_2(L_n))$ rounds.
On every graph, $G$, with $n$ nodes, 
Lemma~\ref{greedyMISuppermu2} says that the round complexity of the Greedy MIS Algorithm  is at most
$\max\{\mu_2(S) +1\ |\ \textrm{ S is a component of } G\} \leq
\mu_2(G)+1 \leq \mu_1(G) + 1\ = n+1$.
Hence, it is asymptotically optimal with respect to $\mu_1$ and $\mu_2$.
\end{proof}

Korman, Sereni and Viennot~\cite{KSV13} give a method for constructing algorithms that
are uniform with respect to $n$ and $\Delta$ from ones that
are not, without increasing the running time by more than a constant factor.
Their idea is to repeatedly execute the original algorithm with increasingly larger guesses for $n$ (and $\Delta$),
but storing the values of the output variables as a tentative solution,
rather than outputting them.
After each repetition, they apply 
a pruning algorithm,
output the extendable partial solution it produces, 
and perform the next repetition on the remainder of the graph, if it is not empty.
However, algorithms that depend on the nodes having distinct identifiers and whose round complexity depends on $d$
remain nonuniform with respect to $d$.
In particular, they obtain a coloring algorithm running in $O(\Delta + \log^*d)$ rounds that is uniform with respect to $\Delta$
and an $O(\log^3 n \log d)$ round matching algorithm that is uniform with respect to $n$.
These algorithms cannot be used to obtain good degradation, since they are not measure-uniform.

\section{Templates for Graph Algorithms with Predictions}
\label{template}
Starting with the early work by Lykouris and
Vassilvitskii~\cite{LV18,LV21} and by Kumar, Purohit, and Svitkina~\cite{KPS18},
an online algorithm that is consistent and has a good
competitive ratio as a function of the prediction error
can sometimes be combined with another online algorithm that is robust
to obtain an online algorithm that is both consistent and robust.
The paper by Kumar, Purohit, and Svitkina
was applied to a nonclairvoyant scheduling
problem, and this was generalized to a framework covering 
additional scheduling problems by Lindermayr and Megow~\cite{LM22}.

Using similar strategies, we develop four general templates
that can be used to obtain 
distributed graph algorithms with predictions. 
The templates all start with a \good initialization algorithm to ensure consistency,
followed by a distributed graph algorithm without predictions, or two that are
combined in some way.
Each template requires certain properties of the algorithmic components they use
and promises certain properties for the resulting algorithm with predictions.
Following the result for each template, we briefly discuss the necessary
properties and present relevant comparisons to earlier templates.

The templates all use a \emph{reference algorithm}, $R$, and ensure that the worst-case round complexity of the resulting algorithm
is within a constant factor of the worst-case round complexity of $R$ (so it is robust with respect to $R$).
Three of them also use a measure-uniform algorithm, $U$, to obtain good degradation.
These algorithms are used mostly in a black box manner, although they may be modified slightly.
For example, an algorithm could be interrupted after a given number of rounds.
To make switching between algorithms work smoothly, we assume that, prior to terminating,
nodes
inform their active neighbors about their output values.
Further details are discussed with the relevant templates.

For each template, we analyze the round complexities of the
resulting algorithm with predictions
as a function of an error measure.
We also 
present an example of an algorithm with predictions for the MIS problem obtained using that template
and determine its round complexity as a function of $\eta_1$ or $\eta_2$.
Since $\eta_2 \leq \eta_1$, an upper bound as a function of $\eta_2$ also implies an upper bound
as a function of $\eta_1$.
Throughout this section, we assume that the round complexities of $R$ and $U$ 
depend on some subset of $\{n,\Delta,d\}$ and are nondecreasing functions of these parameters.
In Section~\ref{black-white},
we give examples of graph algorithms with predictions
for the error measure $\eta_{bw}$, constructed
using our templates (with minor modifications).

\subsection{The Simple Template}
\label{section-simple-template}

Using the first template, we obtain an algorithm with predictions by first running a \good initialization algorithm and then running
any reference algorithm, $R$, without predictions.
The properties of the resulting algorithm will depend on the properties of~$R$. 
Note that the maximum of the degrees of the nodes in an error component is always less than the number of nodes in that error component.

\begin{algorithm}
\caption{Simple Template}\label{simple_template}
\begin{algorithmic}[0]
\State Run a \good initialization algorithm, $B$
\State Run a reference algorithm, $R$.
\end{algorithmic}
\end{algorithm}

\begin{observation}
Given a \good initialization algorithm, $B$, with round complexity $c(n)$ and a reference algorithm, $R$,
with round complexity at most $r(\mu)$, which is uniform with respect to $\mu$,
the algorithm with predictions obtained using the Simple Template (Algorithm~\ref{simple_template}) has consistency $c(n)$, 
has round complexity $c(n) + r(\eta)$,
where $\eta$ is the maximum of $\mu$ over all error components, 
and, hence,
is $r(\eta)$-degrading.
\end{observation}

One example is to use the MIS Initialization Algorithm as $B$ and the  Greedy MIS Algorithm as $R$ in the Simple Template.
The resulting MIS algorithm with predictions has consistency 3.
By Lemma~\ref{greedyMISuppermu1}, the round complexity of the resulting MIS algorithm with predictions is at most $\miseta +3$
and, by Lemma~\ref{greedyMISuppermu2}, it is at most $\eta_2 + 4$.
Thus, it is both $\miseta$-degrading and $\eta_2$-degrading.
It is also robust with respect to the Greedy MIS Algorithm.
However, the worst-case round complexity of this reference algorithm is large, so the resulting algorithm with predictions is not very good.

Another example of the Simple Template uses a different reference algorithm~\cite{KSV13}. It computes a Maximal Independent Set in $O(\Delta+ \log^*d)$
rounds, its round complexity does not depend on $n$, and  it is uniform with respect to both $\Delta$ and $d$.
Essentially, it uses successively significantly larger estimates for the values of these parameters
to make previously known maximal independent set algorithms uniform with respect to these parameters.
The MIS algorithm with predictions that results from using the Simple Template with this algorithm as $R$ and
the MIS Initialization Algorithm as $B$
is consistent and  has round complexity $O(\Delta' + \log^*d)$, where $\Delta'$ is the maximum of the degrees of 
the nodes in the error components.

Since the maximum degree  of the nodes in
 an error component is less than the number of nodes in that error component,
 the round complexity is $O(\eta_1 + \log^*d)$. 
Hence, by definition,  it is not $O(\eta_1)$-degrading or $O(\eta_2)$-degrading.
Note that any deterministic algorithm for the MIS problem on graphs with $n$ nodes and maximum degree $\Delta$ requires 
$\Omega(\min\{ \Delta + \log^*n, (\log n)/\log\log n\} )$ rounds  in the LOCAL model~\cite{BBHORS21}.
If $d \in n^{O(1)}$,
as is often assumed,
then this algorithm with predictions is robust for $\Delta \in O(\log n/\log\log n)$,
and it is close to being $O(\eta_1)$-degrading, since $\log^* d$ is almost constant.

\paragraph{Summary.} The Simple Template ensures consistency.
If  the reference algorithm, $R$, is measure-uniform with respect to $\mu$
and has round complexity 
 $f(\mu)$
then the resulting algorithm is $f(\eta)$-degrading.
When the round complexity of  $R$ is an asymptotically increasing function of $d$, the size of
the domain from which node identifiers are chosen, the resulting algorithm is neither $f(\eta_1)$-degrading
nor $f(\eta_2)$-degrading for any function $f$.
If the round complexity of $R$ does not depend on $d$,
Lemma~\ref{css-mis} says that the worst-case round complexity
of $R$ on a component is at least linear in the size of the component.
Since there are MIS algorithms with polylogarithmic round complexity, in this case, the resulting algorithm is not robust.

\subsection{Consecutive Template}
\label{consec}

Suppose that (an upper bound on) the round complexity, $r(n,\Delta,d)$, of a reference algorithm, $R$,
depends on parameters of the graph that are known to all  nodes.
Then the nodes can compute the value of this function
for the entire graph.
When combined with a measure-uniform algorithm, $U$, with 
$f(\mu)$ round complexity,
we can obtain an algorithm with predictions that is 
$2f(\eta)$-degrading,
where $\eta$ is the maximum of $\mu$ over all error components,
and is  robust with respect to the  reference algorithm.
The idea is to run the measure-uniform algorithm following a \good
initialization algorithm for $r(n,\Delta,d)$ rounds,
to handle instances with small prediction error
and then run
the reference algorithm on the remaining graph. 

If the measure-uniform algorithm does not output a complete solution within the predetermined number of rounds,
it may be necessary to output values at some additional nodes to ensure that, when the reference algorithm begins (on the remaining part of the graph),
the partial solution that has been output is extendable.
A \emph{clean-up} algorithm, $C$, extends a partial solution by having a (possibly empty) set of 
nodes
output values, so that the resulting partial solution is extendable. 
Such an algorithm can generally be defined from a \good initialization algorithm for the same problem, provided that  
nodes
inform their active neighbors of the values they are going to output 
each round. 
We use $c'(n)$ to denote its round complexity, since it might be different than the round complexity, $c(n)$ of the initialization algorithm.

For the MIS problem, a clean-up algorithm can be performed in one round: every active 
node
which has a neighbor that output 1
simply outputs 0 after informing its active neighbors that it will output 0.

A clean-up algorithm is similar to a pruning algorithm~\cite{KSV13}.
However, a pruning algorithm outputs an extendable partial solution that is part of a given tentative solution (which may contain some default values),
whereas a clean-up algorithm is given a partial solution and outputs additional values, if necessary, to make the resulting partial solution extendable.

\begin{algorithm}
\caption{Consecutive Template}\label{consecutive_template}
\begin{algorithmic}[0]
\State Run a \good initialization algorithm, $B$
\State Run a measure-uniform algorithm, $U$, for $r(n,\Delta,d) + c'(n)$ rounds
\State Run a clean-up algorithm, $C$
\State Run a reference algorithm, $R$
\end{algorithmic}
\end{algorithm}

\begin{lemma}
Given a reference algorithm, $R$, with round complexity at most $r(n,\Delta,d)$, a \good initialization algorithm, $B$, with round complexity $c(n) \in O(r(n,\Delta,d))$,
a measure-uniform algorithm, $U$ with round complexity 
$f(\mu)$,
and a clean-up algorithm, $C$, with round complexity $c'(n) \in O(r(n,\Delta,d))$,
the algorithm with predictions obtained using the Consecutive Template (Algorithm~\ref{consecutive_template}) has consistency $c(n)$, is 
$2f(\eta)$-degrading,
where $\eta$ is the maximum of $\mu$ over all error components,
and
is robust with respect to $R$. 
\label{consec-lem}
\end{lemma}
  
\begin{proof}
If $f(\eta) \leq r(n,\Delta,d) + c'(n)$, then 
the resulting algorithm terminates while running $U$, so
the number of rounds performed is at most $c(n)+ f(\eta) \in O(r(n,\Delta,d))$.
If $f(\eta) > r(n,\Delta,d)+ c'(n)$, then 
$R$ is performed after running  algorithms $B$, $U$, and $C$, so
the number of rounds performed is at most $c(n) +  2 r(n,\Delta,d) + 2c'(n)  <  2f(\eta) + c(n)$.
In either case, the number of rounds performed is at most $2f(\eta) + c(n)$, so the algorithm is 
$2f(\eta)$-degrading.
Since the number of rounds performed is also in $O(r(n,\Delta,d))$, the algorithm is robust with respect to $R$.
\end{proof}

For example, consider the 
MIS algorithm by Ghaffari and Grunau~\cite{GG24}
that has $\tilde{O}(\log^{5/3} n)$ round complexity. It depends on 
each node
knowing an upper bound $N \in n^{O(1)}$ on the number of nodes
in the graph and assumes that $d \in n^{O(1)}$.
A distributed graph algorithm with predictions for the MIS problem that is consistent, robust with respect to this algorithm,
and 
both $2\eta_1$-degrading and $2\eta_2$-degrading can be obtained from the Consecutive Template using their algorithm as the reference algorithm, $R$, the MIS Initialization Algorithm as the initialization algorithm,  $B$, the Greedy MIS Algorithm as the 
measure-uniform algorithm, $U$, and the one round clean-up algorithm described above as $C$.

\paragraph{Summary.}
The Consecutive Template ensures consistency.
In the worst case, it causes the round complexity of the resulting algorithm to be twice
the round complexity of the reference algorithm, $R$, 
and
its degradation
function to be twice the round complexity, $f(\mu)$, of the measure-uniform algorithm, $U$.
Provided $R$ is asymptotically optimal, it also ensures robustness.
In addition, if $f(\mu)$
does not grow too quickly, then the resulting algorithm is smooth with respect to $R$.
In contrast to the Simple Template, we can obtain both robustness and linear degradation simultaneously for MIS.
However, it requires that all nodes are able to compute good upper bounds on the round complexity,  $r(n,\Delta,d)$, of $R$ and the round complexity, $c'(n)$, of the clean-up algorithm $C$.

\subsection{Interleaved Template}
\label{interleaved}
Consider a reference algorithm, $R$, with round complexity $r(n,\Delta,d)$ that  can be divided into $m(n,\Delta,d)$ phases with a good known upper bound,
$r_i(n,\Delta,d)$,
on the round complexity of phase $i$,
such that the sum of $r_i(n,\Delta,d) $ over all phases is in $O(r(n,\Delta,d))$.
Also consider
a  measure-uniform algorithm, $U$, with round complexity 
$f(\mu)$
that can be divided into phases with these round complexities.
Suppose that,
at the end of each phase of these algorithms, the partial solution that has been
computed is extendable.
In this case, using the Interleaved Template, we can obtain an algorithm with predictions
by running a \good initialization algorithm, $B$, with round complexity $c(n) \in O(r(n,\Delta,d))$, and then running
the phases of the measure-uniform algorithm and the reference algorithm in an interleaved manner.
The first phase of $U$ is run first to handle instances with small prediction error.
During each round, all active nodes should be performing the same phase of the same algorithm.
Thus, unless a node has terminated, it should wait until the number of rounds that has elapsed
in a phase is the known upper bound for that phase, before starting the next phase.

\begin{algorithm}
\caption{Interleaved Template}\label{interleaved_template}
\begin{algorithmic}[0]
\State Run a \good initialization algorithm, $B$
\For{$i = 1, 2, 3, \ldots, m(n,\Delta,d)$}
\State Run phase $i$
of the measure-uniform algorithm, $U$, for $r_i(n,\Delta,d)$ rounds
\State Run phase $i$ of the reference algorithm, $R$,  for $r_i(n,\Delta,d)$ rounds
\EndFor
\end{algorithmic}
\end{algorithm}

\begin{lemma}
  \label{interleavedlemma}
Suppose we have
a reference algorithm, $R$, with round complexity $r(n,\Delta,d)$, and a  measure-uniform algorithm, $U$, with round complexity 
$f(\mu$),
both of which  can be divided into $m(n,\Delta,d)$ phases, where the round complexity of phase $i$ is at most $r_i(n,\Delta,d)$, 
where the sum of $r_i(n,\Delta,d)$ over all phases is in $O(r(n,\Delta,d))$.
Also suppose that  the partial solution at the end of each phase of these algorithms is extendable.
Given
 a \good initialization algorithm, $B$, with round complexity $c(n) \in O(r(n,\Delta,d))$,
 the algorithm with predictions obtained using the Interleaved Template (Algorithm~\ref{interleaved_template}) has consistency $c(n)$, is 
$2f(\eta)$-degrading, and is
robust with respect to $R$.
\label{interleave-lem}
\end{lemma}

\begin{proof}
Since the error, $\eta$, does not increase when errors in the predictions are removed
and the round complexities of $R$ and $U$ are assumed to be nondecreasing functions,
it follows that interleaving these algorithms does not increase the maximum number of rounds each performs. 

We now determine an upper bound on the round complexity of the resulting algorithm with predictions.
If $f(\eta) \leq r_1(n,\Delta,d)$, then the number of rounds performed
is at most $c(n) + f(\eta) \leq c(n) +  r_1(n,\Delta,d) \in O(r(n,\Delta,d))$.
If $\sum_{i=1}^{j} r_i(n,\Delta,d) < f(\eta) \leq \sum_{i=1}^{j+1}r_i(n,\Delta,d))$, where $1 \leq  j < m(n,\Delta,d)$,
then part of $U$ is performed and at most $j$ phases of $R$ are performed,
for a total of at most $c(n) + f(\eta) + \sum_{i=1}^{j} r_i(n,\Delta,d) < c(n) + 2 f(\eta)
\leq c(n) + 2 \sum_{i=1}^{j+1}r_i(n,\Delta,d))  \in O(r(n,\Delta,d))$ rounds.
If 
$\sum_{i=1}^{m(n,\Delta,d)} r_i(n,\Delta,d) < f(\eta)$, then it is possible that all $m(n,\Delta,d)$ phases of $R$ are performed,
for a total of at most $c(n) +  2\sum_{i=1}^{m(n,\Delta,d)} r_i(n,\Delta,d)  < c(n) + 2 f(\eta)$ rounds.
Also note that  $c(n) +   2\sum_{i=1}^{m(n,\Delta,d)} r_i(n,\Delta,d) \in O(r(n,\Delta,d))$.

In all cases, the number of rounds performed is at most $c(n) +2f (\eta)$,
so the algorithm is $2f(\eta)$-degrading.
Since the number of rounds performed is also in $O(r(n,\Delta,d))$, the algorithm
is robust with respect to~$R$.
\end{proof}

Next, we present an example of how this template can be applied.
We start with a discussion of the reference algorithm, in some detail, in order to
define phase lengths.
Ghaffari, Grunau, Haeupler, Ilchi and Rozho\v{n}~\cite{GGHIR23} present an 
$\tilde{O}(\log n \log d)$ round clustering algorithm.
Given a graph with $n$ nodes, it computes
a collection of disjoint sets of nodes whose union has size at least $n/2$
such that nodes in different sets are not adjacent in the graph
and the subgraph induced by each set has diameter $O(\log n\log\log n)$.
For each of these induced subgraphs, it also computes a breadth-first search tree.
This tree can be used to compute a maximal independent set of this induced subgraph
in $O(\log n\log\log n)$ rounds.
Consider the reference algorithm, $R$, obtained by repeatedly performing phases that
consist of running the clustering algorithm,
computing a maximal independent set of the subgraph induced by each set,
and then performing the one round clean-up algorithm.
Since the number of nodes remaining in the graph at the end of phase $i$ is at most
$n/2^i$, it is easy for these nodes to compute an upper bound $r_{i+1} \in \tilde{O}(\log(n/2^i)\log d)$ on the number
of rounds needed for the next phase.
It is not necessary that they compute the number of phases that will be performed.

\begin{corollary}
  Using the Interleaved Template with the MIS Initialization Algorithm as the initialization
  algorithm, $B$,
  the Greedy MIS Algorithm as the measure-uniform algorithm, $U$,
  and the reference algorithm, $R$,  with phase lengths as
  just described, gives an algorithm with predictions for MIS that has
  consistency~3 and round complexity $3+2\min\{ \eta_1, \tilde{O}(\log \eta_1\log n \log d) \}$.
  It is both $2\eta_1$-degrading and $2\eta_2$-degrading, and it is robust and smooth with respect to the reference algorithm $R$.
\end{corollary}
\begin{proof}
  The MIS Initialization Algorithm consists of three rounds.
The Greedy MIS Algorithm has round complexity bounded above by both 
$\mu_1$ and $\mu_2 + 1$.
The partial solution it outputs at the end of every even numbered round is extendable, so
we choose $r_i$ to be even for $1 \leq i \leq \log n$ in
the Interleaved Template.
By Lemma~\ref{interleavedlemma},  the resulting algorithm has consistency 3, is both $2\eta_1$-degrading and $2\eta_2$-degrading, and
is robust with respect to $R$. Hence it is also smooth with respect to $R$.

Since the number of active nodes in each error component decreases by at
least a factor of two for each phase of $R$, at most $\log \eta_1$ iterations are performed.
Hence,  the resulting algorithm has round complexity at most
 $c(n) +   2\sum_{i=1}^{\log\eta_1} r_i(n,\Delta,d) \in \tilde{O}(\log \eta_1 \log n \log d)$.
\end{proof}

\paragraph{Summary.} The Interleaved Template ensures consistency.
Like
the Consecutive Template, 
it causes the round complexity of the resulting algorithm to be twice
the round complexity of the reference algorithm, $R$,
and its degradation 
 function to be twice the round complexity, $f(\mu)$, of the measure-uniform algorithm, $U$.
Provided $R$ is asymptotically optimal, it also ensures robustness.
In addition, if $f(\mu)$,
does not grow too quickly, then the resulting algorithm is smooth with respect to $R$.
In contrast to the Simple Template, we can obtain both robustness and linear degradation simultaneously for MIS.
It requires that both $R$ and $U$ can be divided into phases with the same upper bound on their round complexities, which can be computed by all nodes.
Note that the Greedy MIS Algorithm and the other measure-uniform algorithms described in Section \ref{section-other} all work in phases consisting of a small constant number of rounds.
Hence, they can easily be used as $U$ in the Interleaved Template.
As compared to the Consecutive Template, the resulting algorithm may terminate earlier, since the reference algorithm may reduce the value of the measure of
the remaining problem significantly each round
or the reference algorithm might terminate earlier than the known upper bound, possibly due to good predictions.

\subsection{Parallel Template}
\label{parallel-template}

Using a measure-uniform algorithm with round complexity 
$f(\mu)$,
in either the Consecutive or the Interleaved Template gives an algorithm with predictions that is $2f(\eta)$-degrading,
where $\eta$ is the maximum of $\mu$ over all error components.
We show that, under certain conditions, the factor of 2 can be removed.

The solutions to some graph problems remain solutions when nodes are removed from the graph.
For example, a coloring of a graph is still a coloring of any induced subgraph.
Similarly, if $M$ is an independent set 
 of a graph $G$ and $G' = (V',E')$ is an induced subgraph of $G$,
then $M\cap V'$ is an independent set 
of $G'$.
However, a maximal independent set 
may no longer be maximal
when nodes are removed. For example, just considering two nodes with an edge between them, a set consisting of one of these 
nodes is a maximal independent set, but, if that node is removed, the empty set is not a maximal independent
set of the remaining graph.

Consider a graph problem whose solutions remain solutions when nodes
are removed.  An algorithm for such a problem is \emph{fault-tolerant}
if it runs in a model where nodes can crash and, at any point, the
partial solution output by the nodes that have not crashed is a
partial solution to the problem on the subgraph induced by these
nodes.

Suppose that a reference algorithm, $R$, has two
parts, the first of which is fault-tolerant.
The second part can be empty, if the entire reference algorithm is fault-tolerant.
The idea of the \para is to run the first 
part of $R$ in parallel with a measure-uniform algorithm, $U$.
Even though the model assumes that 
nodes
are nonfaulty, a  node
that terminates while it is performing the measure-uniform algorithm
can be treated as 
if it had crashed
by the reference algorithm.
Nodes
must know a good upper bound, $r_1(n,\Delta,d)$ on the round complexity of the first part
of the reference
algorithm.
To avoid interference between the two algorithms,
if the first part of the reference algorithm outputs values, it is modified so that
 these  values are, instead, stored locally by nodes.
While performing the measure-uniform algorithm, nodes still output values.
A node that terminates while performing the measure-uniform algorithm
no longer participates in the reference algorithm and it is treated as if it has crashed.

After the first part
 of the reference algorithm is known to have completed,
the partial solution produced by the measure-uniform
algorithm is made extendable by running a clean-up algorithm.
The locally stored outputs, if any, computed by the first part
of the reference algorithm
are output. Then the second part
of the reference algorithm is performed.

Although a measure-uniform  algorithm with round complexity $f(\mu)$ will terminate
within $f(\eta)$ rounds 
when $f(\eta) \leq r_1(n,\Delta,d)$,
it is not clear
that it will have made
enough progress on the problem instance
if it does not terminate.
So, even if part 2 of $R$ has round complexity at most $f(\mu)$,
it only follows that the round complexity of the resulting algorithm with predictions
is at most $c(n)+ r_1(n,\Delta,d) + f(\eta) \leq 2f(\eta)$.
To ensure that its round complexity is at most $c(n) + f(\eta) +O(n)$ and, hence, that it is $f(\eta)$-degrading,
we introduce another property.

A measure-uniform algorithm, $U$, with round complexity at most $f(\mu)$ \emph{makes steady progress with respect to} $\mu$
if, for every connected 
graph $G$ and every $1 \leq r < f(\mu(G))$, each component, $S$, of the subgraph  induced by the set of active nodes at the 
end of round $r$ when $U$ is performed on $G$ has the property that
$r + f(\mu(S)) \leq f(\mu(G)) + O(1)$.
Since the round complexity of the Greedy MIS Algorithm is at most $\mu_1$ and, in every round, at least one node in every component
terminates, it makes steady progress with respect to $\mu_1$.
Since the round complexity of the Greedy MIS Algorithm is at most $\mu_2+1$ and in every other round the Greedy MIS algorithm adds at least
  one node to the maximal independent set, the size of the remaining independent
  set and the remaining vertex cover are decreased by at least $1$ every 
  second round (as explained in the proof of Lemma~\ref{greedyMISuppermu2}).
  Thus, $\mu_2$ is decreased by at least two every other round, so the
Greedy MIS algorithm also makes steady progress with respect to $\mu_2$.

An alternative way to consider steady progress is as follows.
Since $f$ gives an upper bound on the round complexity 
of a measure-uniform algorithm as a function
of the work to be done, as quantified by our measure $\mu$,
$f^{-1}$ gives a lower bound on the work accomplished 
by the measure-uniform algorithm (again, as quantified
by $\mu$) as a function of the round complexity.
Suppose that the algorithm is performed on a connected graph, $G$, and let $G_r$ denote
the subgraph induced by the set of active  nodes at the end of round $r$.
The remaining amount of work to be done is $\max\{ \mu(S)\ | \ S \mbox{ is a component of } G_r\}$
and $\mu(G) - \max\{ \mu(S)\ | \ S \mbox{ is a component of } G_r\}$ is the amount of work
 that has been accomplished.
 Hence, $f^{-1}(r) \leq \mu(G) - \mu(S)$ for every component $S$ of $G_r$.
 If $f$ is  superadditive,
 then this condition implies that the algorithm makes steady progress with respect to $\mu$:
Because $f$ is assumed to be nondecreasing, $r \leq f(\mu(G) - \mu(S))$ for every component $S$ of $G_r$.
So, by superadditivity,
$r + f(\mu(S)) \leq f(\mu(G) - \mu(S)) + f(\mu(S))
\leq f(\mu(G) - \mu(S) + \mu(S)) = f(\mu(G))$.

\begin{algorithm}
\caption{Parallel Template}\label{parallel_template}
\begin{algorithmic}[0]
  \State Run a \good initialization algorithm, $B$
  \For{$r_1(n,\Delta,d)$ rounds, in parallel}
     \State Run the next round of the measure-uniform algorithm, $U$
     \State Run the next round of part 1 of the reference algorithm, $R$,
     \State storing any outputs locally
  \EndFor
  \State Run the clean-up algorithm, $C$
  \State Output any locally stored outputs
  \State Run part 2 of $R$
\end{algorithmic}
\end{algorithm}

\begin{lemma}
\label{parallel}
Suppose we have a reference algorithm, $R$, with round complexity $r(n,\Delta,d)$,
that has a fault-tolerant first part
with round complexity 
at most  $r_1(n,\Delta,d) \leq r(n,\Delta,d)$, which
is known to all nodes,
a measure-uniform algorithm, $U$,
with  round complexity at most $f(\mu)$,
a reasonable initialization algorithm, $B$, with round complexity $c(n) \in O(r(n,\Delta,d)$, and
a clean-up algorithm, $C$,  with round complexity $c'(n)\in O(r(n,\Delta,d))$.
If either the round complexity of $C$ followed by the second part of $R$ is constant
or $U$ makes steady progress with respect to $\mu$ and 
the round complexity of $C$ followed by the second part of $R$ is at most $f(\mu) +O(1)$, 
then
the algorithm with predictions obtained using the Parallel Template (Algorithm~\ref{parallel_template}) has consistency $c(n)$, 
is robust with respect to $R$, and
is $f(\eta)$-degrading,  where $\eta$ is the maximum of $\mu$ over all error components.
\end{lemma}

\begin{proof}
The resulting algorithm has round complexity at most $c(n)+c'(n)+r(n,\Delta,d) \in O(r(n,\Delta,d))$, so it is robust with respect to $R$.

If the resulting algorithm terminates during the for loop, then the number of rounds performed
is at most $c(n) + f(\eta)$.

Otherwise,  $f(\eta) >  r_1(n,\Delta,d)$. 
If the round complexity of $C$ followed by the second part of $R$ is  constant,
then the number of rounds performed is at most $c(n) + r_1(n,\Delta,d) +O(1) < c(n) + f(\eta)  +O(1)$.

Now suppose that $U$ makes steady progress with respect to $\mu$
and that
the round complexity of $C$ followed by the second part of $R$ is at most $f(\mu) +O(1)$. 
Let $G'$ be the subgraph induced by the active nodes at the end of the initialization algorithm, $B$, and let
$G''$ be the subgraph induced by the active nodes at the end of the for loop.
Then the  number of rounds performed is at most $c(n) + r_1(n,\Delta,d) + 
 \max\{f(\mu(S''))\ |\ S'' \mbox{ is a component of } G''\} +O(1)$.
Note that  $U$ runs independently on each component of $G'$, so
each component $S''$ of $G''$ is
an induced subgraph of some component, $S'$, of $G'$.
Since $U$ performs  $r_1(n,\Delta,d)$ rounds during the for loop and  makes steady progress
with respect to $\mu$,
it follows that $r_1(n,\Delta,d) + f(\mu(S'')) \leq f(\mu(S')) + O(1)$.
Hence,
 the number of rounds performed is at most
 \begin{eqnarray*}
& & c(n) + r_1(n,\Delta,d) + 
 \max\{f(\mu(S''))\ |\ S'' \mbox{ is a component of } G''\}
  + O(1)\\
&  \leq  & c(n) +  \max\{f(\mu(S'))\ |\ S' \mbox{ is a component of } G'\} +O(1).
 \end{eqnarray*}
Since $f$ is nondecreasing, this is  at most
$c(n) + f(\eta) + O(1)$.

Hence, in all
cases, the number of rounds performed is at most $c(n) + f(\eta)  +O(1)$.
It follows that the algorithm with predictions is $f(\eta)$-degrading.
\end{proof}

As an example of how this template is used, we let $R$ be the reference algorithm
whose first part is a fault-tolerant  ($\Delta$+1)-Vertex Coloring algorithm~\cite{L92,BE13,BEG22} 
with round complexity $r_1(n,d,\Delta)\in O(\Delta+\log^*d)$
and whose second part is
a simple greedy algorithm~\cite{BE13} that produces a maximal independent set from this coloring in 
$\Delta$ rounds.
The second part considers the colors
1 to $\Delta$
one at a time:
In round $1 \leq i \leq \Delta$, each 
active
node
that has color $i$ and has not received 1 from any of its neighbors
sends the message 1
to each of its neighbors (indicating that it is  in the maximal independent set), outputs 1, and terminates.
If $1 < i < \Delta$ and
an active node received
1 from at least one neighbor in 
round $i-1$,
it sends the message 0 to each of its neighbors (indicating that it is not in the maximal independent set)
outputs 0, and terminates.
Note that, after round $i < \Delta$, all active nodes have colors greater than $i$
and each node knows which of its neighbors are active.
In round $\Delta$, each active node outputs 0 and terminates if it received 1 in this round or the previous round.
Each other node with color $\Delta+1$ outputs 1 and terminates in round $\Delta$, but does not send any messages.
The correctness of this algorithm relies on the fact that  the set of nodes with the same color is an independent set.

If $\Delta$ is constant, then the second part of $R$ has constant round complexity.
Otherwise, to ensure that the resulting algorithm is both $\miseta$-degrading and $\eta_2$-degrading, we can combine 
it with the Greedy  MIS Algorithm,
as follows: 
In addition,
in round $1 \leq i < \Delta$, each active node with color 
greater than $i$ that has not received 1 from any of its neighbors,
has no active neighbor with color $i$,
and whose identifier is larger than the identifiers of all its active neighbors
sends 1 to each of its neighbors
(indicating that it is in the maximal independent set), 
outputs 1, and terminates.
Nodes that receive this message also send the message 0 to each of their neighbors, output 0, and terminate.

\begin{corollary}
  \label{cor-parallel}
  Using the Parallel Template with the MIS Initialization Algorithm as the initialization
  algorithm, $B$,
  the Greedy MIS Algorithm as the measure-uniform algorithm, $U$,
 and the two-part reference algorithm, $R$, 
  just described, gives an algorithm with predictions for MIS that has
  consistency~3 and round complexity $\min\{ \eta_2+4, O(\Delta+\log^*d) \}$.
It is both $\eta_1$-degrading and $\eta_2$-degrading.
When $\Delta \in O((\log n)/\log\log n)$, it is
 robust and smooth.
\end{corollary}

\begin{proof}
The MIS Initialization Algorithm consists of 3 rounds.
Since part 1 of $R$ performs $O(\Delta+\log^*d)$ rounds and part 2 of $R$ performs at most $\Delta$ rounds,
the resulting algorithm with predictions has round complexity $O(\Delta+\log^*d)$. 

The partial solution output by the Greedy MIS Algorithm at the end of every even numbered round is extendable,
so we choose $r_1(n,d,\Delta)$ to be even. Then no clean-up algorithm is needed.
So, when we later apply Lemma~\ref{parallel}, $c'(n)=0$.

Consider any component $S$ of the graph induced by the active nodes at the beginning of round $i$ of  part 2 of $R$.
All nodes in $S$ that received 1 in the previous round output 0 and terminate.
Suppose no nodes in $S$ received 1 in the previous round. 
Then all nodes of color $i$ in $S$ are added to the maximal independent set in round $i$.
If $S$ contains no nodes of color $i$, then the node in $S$ with the largest identifier
has no active neighbor with color $i$, so it is added to the maximal independent set in round $i$.
All nodes added to the maximal independent set in round $i$ send 1 to all their neighbors and terminate.
Thus, at least one node in $S$ terminates during round $i$. If no node in $S$ is added to the maximal independent set
during round $i$, then no node in $S$ receives 1 in round $i$, so at least one node is added to the maximal independent set
in every component of the subgraph of $S$ induced by the active nodes at the beginning of round $i+1$.
Therefore, a node is added to the independent set at least every other round.

As explained in the proof of Lemma~\ref{greedyMISuppermu2}, it follows that
the size of the remaining independent
set and the remaining vertex cover are decreased by at least 1 every 
other round, so
the round complexity  of the second part of $R$ is bounded above by $\mu_2(S) + 1$
when performed on a component, $S$, of the graph induced by the active nodes at the beginning of part 2 of $R$.
Since the Greedy MIS algorithm also 
has round complexity at most $\mu_2 +1$ and makes steady progress with respect to $\mu_2$, 
Lemma~\ref{parallel} implies that the resulting algorithm with predictions has round complexity at most $3 + \eta_2 + 1$.
Hence, it is $\eta_2$-degrading. 
In addition,
because $\eta_2 \leq \miseta$, 
this algorithm is
$\eta_1$-degrading.
\end{proof}

\paragraph{Summary.}
The Parallel Template ensures consistency. Unlike the
Consecutive and Interleaved Templates, 
it does not cause the round complexity of the resulting algorithm to be twice the round complexity of the reference algorithm, $R$.
Provided $R$ is asymptotically optimal, it also ensures robustness.
It requires that the first part of $R$ is fault tolerant and that all nodes know a good upper bound on the round complexity of this part.
If the clean-up algorithm followed by the second part of $R$ has constant round complexity,
then the resulting algorithm is $f(\eta)$-degrading, where $f(\mu)$ is an upper bound on the round complexity of 
the measure-uniform algorithm, $U$.
If the clean-up algorithm followed by the second part of $R$ has round complexity bounded above by $f(\mu)$
and $U$ makes steady progress with respect to $\mu$, then the resulting algorithm is also $f(\eta)$-degrading.
In both cases, provided $f$ does not grow too quickly, the algorithm is smooth with respect to $R$.

\section{Applying Templates to Other Problems}
\label{section-other}

We now consider three additional graph problems and extendable partial
solutions for them.
The  Maximal Matching, $(\Delta+1)$-Vertex Coloring,
and $(2\Delta-1)$-Edge Coloring problems are known to be very closely related
to the MIS problem. 
Algorithms for MIS often lead to similar 
algorithms for the other problems.
For all three problems,
each  node
is given a prediction for its output.

We 
give
definitions of 
extendable partial
solutions, initialization algorithms, clean-up algorithms,
and error measures
for each of these problems.
Based on that, one can then choose one's favorite algorithm for 
the problem and use that as the reference algorithm.
Which template(s) from Section~\ref{template} %that
can be used will, of course, depend on the properties
of the reference algorithm that is chosen.

\subsection{Maximal Matching}

A \emph{matching} of a graph is a subset $M$ of its edges such that no two edges in $M$
are incident to the same node. It is \emph{maximal} if it is contained in no other matching.
In the \emph{Maximal Matching problem},
the output $y_i$ of 
node
$p_i$ 
 is either the neighboring node to which it is matched
 or $\bot$, indicating that it is unmatched.
 When all 
 nodes
 have terminated, $y_i = j$ if and only if $y_j = i$
 and we say the nodes $p_i$ and $p_j$ are \emph{matched}.
 If $y_i=\bot$, then all nodes adjacent to $p_i$  have been matched.
Given any matching $M$ of a graph $G$, consider the set of nodes $N$ that are incident to some edge in $M$
and the set of nodes $N'$ that are not in $N$, but all of whose neighbors are in $N$.
If $M'$ is any maximal matching of the graph
 obtained by removing $N \cup N'$ from $G$, then $M' \cup M$ is a maximal matching of $G$.
 Thus, an extendable partial solution can be obtained by setting $y_i = j$ and $y_j = i$ for all edges $\{p_i,p_j\} \in M$
 and setting $y_i = \bot$ for all nodes $i \in N'$.

The Maximal Matching Base Algorithm proceeds as follows:
Every 
node
sends 
its prediction to each of its neighbors. Then each 
node
$p_i$ with
prediction $x_i = j \neq i$  checks that $x_j = i$.
If so, it informs its other neighbors, outputs the
match, and terminates.
Each 
node
with prediction~$\bot$ outputs $\bot$ and terminates
if it learns that all of its neighbors have been matched.
A 
node
that outputs $\bot$ does not need to inform its neighbors,
since they have all terminated. Thus, this pruning algorithm takes only two rounds.

Another initialization algorithm that is always at least as
  good as the Maximal Matching Base Algorithm has a node output $\bot$, even if its prediction is not $\bot$,
provided all its neighbors are all matched.
Note that this
is a \good initialization algorithm, but it is not a pruning algorithm.

The error components are the components of the
subgraph induced by the nodes that would not have produced output
after running the base algorithm, and $\eta_1$ can be used as the error measure.

In a clean-up algorithm for the Maximal Matching problem,
it suffices that every active node which has been matched to a neighbor informs its other neighbors that it has been matched, outputs the identifier of
the neighbor to which it has been matched, and terminates.

The Maximal Matching problem has a measure-uniform algorithm that is  similar to the one for the Maximal Independent Set problem.
Rounds are divided into groups of three. In the first round of a group,
 if the identifier of an active  
 node
 is larger than
  the identifiers of all of its active neighbors, 
it proposes to 
its active neighbor with the smallest identifier
that they be matched.
Each 
node
that receives a proposal chooses 
the one from the neighbor with the largest identifier
and, in the subsequent  round,
notifies that neighbor that they are matched.
Each matched node informs its active neighbors in the third round and then terminates.
When an active
node
receives a message during the third round, it removes the 
node
from which
it received the message as a neighbor. If it no longer has any active neighbors, it outputs $\bot$ and terminates.
At the end of each group of three rounds, each component of the remaining graph
has at least two nodes. In the next group of three rounds, at least two of them are matched and terminate.
If this component has exactly three nodes, then
the unmatched node also terminates.
 Thus, this algorithm has 
 round complexity at most $3\lfloor s/2 \rfloor$ when performed on a component with $s\geq 2$ nodes.
 When $s= 1$, the node
 knows that it has no active neighbors, so it simply outputs $\bot$ and terminates.
We also show that this measure-uniform algorithm has asymptotically optimal round complexity.

\begin{lemma}
Every deterministic measure-uniform algorithm for computing a maximal matching of an $n$ node line requires at least $(n-3)/2$ rounds,
if its round complexity does  not depend on the domain from which the node identifiers are chosen
or $d$ is a sufficiently large function of $n$.
\end{lemma}

\begin{proof}
To obtain a contradiction,
suppose there is a deterministic measure-uniform algorithm $A$ that always computes a maximal matching of this graph in fewer than $(n-3)/2$ rounds. 
Have the left endpoint of each matched edge output 0, have the right endpoint of each matched edge output 1, and have each unmatched node output 2.
This is an algorithm that colors the nodes of an $n$ node line with 3 colors in fewer than $(n-3)/2$ rounds, contradicting Lemma~\ref{css-vc}.
\end{proof}

\subsection{\boldmath$(\Delta +1)$-Vertex Coloring}
For the \emph{$(\Delta +1)$-Vertex Coloring problem}, the output $y_i$ of node
$p_i$ is a color in $\{1,\ldots, \Delta+1\}$ such that no 
node
has the same color as any of its neighbors.
It  is a special  case of the \emph{list vertex coloring}
problem: each node, $p_i$, has a 
subset of the colors $\{1,\ldots, \Delta+1\}$,
called  its \emph{palette}, which is larger than the degree of its node.
The requirement is for each node to output
a color $y_i$ from its palette, so that no node has output the same color as any of its neighbors.
A partial solution to this  problem 
 consists of the outputs of  a subset of nodes,
each with a color from its palette,
so that no node has output the same color as any of its neighbors.
A partial solution is extendable if the palette of each active node consists of the colors in $\{1,\ldots, \Delta+1\}$ that have not been output by any of its neighbors.
This is because each node
has more colors in its palette than its degree,
so even if all the neighbors of a node have output a color, at least one color remains in its palette.

The $(\Delta+1)$-Vertex Coloring Base Algorithm proceeds as follows:
Every node sends 
its prediction to each of its neighbors. If
its prediction is different from those of all of its neighbors,
it informs all its neighbors, outputs its predicted
color, and terminates.
Each active node that is informed removes this color from its palette.
This algorithm takes two rounds.
There is another initialization algorithm for $(\Delta +1)$-Vertex Coloring,
where, instead, 
a node outputs its predicted color provided
all of its neighbors with the same prediction have smaller identifiers.
This is also a pruning algorithm, and
the extendable partial solution obtained contains the one produced
by the $(\Delta+1)$-Vertex Coloring Base Algorithm, so it is a \good
initialization algorithm.

The error components are the components of the
subgraph induced by the nodes that would not have produced output
after running the base algorithm, and $\eta_1$ can be used as the error measure.

No clean-up algorithm is needed for the $(\Delta +1)$-Vertex Color problem.

A simple measure-uniform algorithm for the $(\Delta+1)$-Vertex Coloring problem
repeatedly has each active node with an
identifier larger than those of all of its active neighbors choose a color
from its palette,
inform its active neighbors of its chosen color, and terminate.
When informed, an active node removes the color from its palette and removes the node from which
it received the message as a neighbor.
In each round, at least one node in each component of the remaining graph terminates,
so the algorithm has round complexity at most $s$ when performed on
a graph with $s$ nodes.
By Lemma~\ref{css-vc} in Section~\ref{sec-css}, this measure-uniform algorithm has asymptotically optimal round complexity.

\subsection{\boldmath$(2\Delta -1)$-Edge Coloring}

For the \emph{$(2\Delta -1)$-Edge Coloring problem}, each node, $p_i$, must output
a different color, $y_{i,e} \in  \{ 1,\ldots,2\Delta-1\}$,
for each edge, $e$, incident to it.
A node can output the colors for different edges in different rounds.
It terminates when it has output an assignment for each of 
its incident edges.
Every pair of neighbors, $p_i$ and $p_j$, must
output the same color for the edge $\{p_i,p_j\}$.
The $(2\Delta -1)$-Edge Coloring problem is a special  case of the \emph{list edge coloring}
problem: 
each edge has a subset of the colors $\{1,\ldots, 2\Delta-1\}$,
which is larger than the number of other edges that are incident to the edge.
This palette is maintained by 
both of its incident nodes.
The requirement is for each node, $p_i$, to output
a color $y_{i,e}$ from the palette for each incident edge, $e=\{p_i,p_j\}$, 
so that $y_{i,e} = y_{j,e}$
and all its incident edges
have different colors.
A partial solution for this problem
on the graph $G=(V,E)$ consists of a subset of edges $E' \subseteq E$ 
such that, for each edge $\{p_i,p_j\}\in E'$, 
nodes
$p_i$ and $p_j$
have output the same color for this edge, for each edge $\{p_i,p_j \}\not\in E'$,
neither node $p_i$ nor $p_j$ has output a color for this edge,
and no node has output the same color for two of its incident edges.
A partial solution $E'$ is extendable if,
for each edge $\{p_i,p_j \}\not\in E'$, nodes $p_i$ and $p_j$
have the same palette for that edge,
which consists of those colors that have not been
output by either node.

The $(2\Delta-1)$-Edge Coloring Base Algorithm proceeds as follows:
Each node sends its predicted color for an edge to its neighbor incident
to that edge, provided that none of its other edges have the same
predicted color. If that neighbor sends the same color, 
they output that color for the edge.
If a node has output a color for each of its incident edges, then it
terminates at the end of round 1.
Otherwise,  it removes the colors it has 
output
from the palettes
for its uncolored edges and, in round 2, sends those colors along its uncolored edges.
A node that receives a color along an edge in round 2 also removes that color from the palette for that edge.
If the prediction is
correct, this pruning algorithm terminates in one round; otherwise it takes at
most two rounds.

For the $(2\Delta-1)$-Edge Coloring problem, the error components are the
components of the subgraph induced by the edges that would remain uncolored
after running the base algorithm, and $\eta_1$ can be used as the error measure.
One could instead use the maximum of the number of edges in the
error components as the error measure. This makes sense since coloring the edges of a graph is equivalent to  coloring the nodes of its line graph.
However, the number of edges in a component with $s$ nodes is
at least $s-1$ and may be much larger.
As mentioned in Section~\ref{intro-error-measures}, when possible, it is preferable to use an error measure that returns smaller values.

A clean-up algorithm for the $(2\Delta -1)$-Edge Color problem could simply have each active process send
any color it has output along its uncolored edges and remove any color that it
receives
along an edge from the palette for that edge.

A measure-uniform algorithm for the $(2\Delta-1)$-Edge Coloring problem is slightly more complicated than for the other three problems.
Each active node keeps track of 
which of the edges adjacent to it and its neighbors have been colored.
This information will be known at the beginning of the measure-uniform algorithm if, 
in round 2 of the initialization algorithm,
each node also sends, along each of its uncolored edges, the identifiers of all other nodes
at the endpoints of these edges.
In each odd-numbered round of the measure-uniform algorithm,
consider each active  node, $p_i$, whose identifier is  larger than the identifiers of the 
nodes
that can be reached from it by following at most two uncolored edges.
For each of its uncolored incident
edges, $e$,
$p_i$ chooses a different color from its palette for that
edge, sends that color to 
the other endpoint of $e$, sets $y_{i,e}$ to that
color and outputs it. Then it terminates.
A node, $p_j$, that receives a color 
$c$ along edge $e=\{p_i,p_j\}$,
sets $y_{j,e}$ to $c$ % that color
and outputs it.
If there are  no longer any uncolored edges incident to $p_j$,
it terminates.
Otherwise, $p_j$ removes $c$
from its palette for each uncolored edge incident to its node.
Then, in the subsequent even-numbered round,
$p_j$ informs each of its other neighbors to remove $c$
from the palette of
their shared edge and that edge $e$ is now colored.
At least one node terminates in every odd-numbered round.
If a component contains exactly two nodes, both 
terminate in round~1.
Thus, this algorithm has 
round complexity at most $2s - 3$ when performed on a graph with $s \geq 2$ nodes.
 A graph with 1 node has no edges, so no output is produced and this algorithm terminates in 0 rounds.
We also show that this measure-uniform algorithm has asymptotically optimal round complexity.

\begin{lemma}
Every deterministic measure-uniform algorithm for coloring the edges of  an $n$ node line with 3 colors requires at least $(n-3)/2$ rounds,
if its round complexity does  not depend on the domain from which the node identifiers are chosen
or $d$ is a sufficiently large function of $n$.
\end{lemma}
\begin{proof}
To obtain a contradiction,
suppose there is a deterministic measure-uniform algorithm $A$ that always colors the edges of this graph using  3 colors in fewer than $(n-3)/2$ rounds. 
Have each node, except the last, output the color of the incident edge to its right. Have the last node output one of the colors different 
from the color of the incident edge to its left. This is an algorithm that colors the nodes of an $n$ node line with 3 colors in fewer than $(n-3)/2$ rounds, 
contradicting Lemma~\ref{css-vc}.
\end{proof}

\section{Revisiting MIS: Black and White Components}
  \label{black-white}

  In this section, we consider 
  MIS algorithms using error measures based on
  the black and white components introduced in Section~\ref{error-measures}.
Recall that a {\em black (white) component} is a component of the graph induced by the nodes that have prediction 1 (0) and were active 
at the end of the base algorithm.
  For general graphs, we use the measure $\eta_{bw}$, which is
  the maximum number of nodes in any black or
  white component after running the MIS Base Algorithm.
  We also consider a different initialization algorithm and error measure for rooted trees.
  
\subsection{An MIS Algorithm with Predictions for General Graphs}
\label{general}
  First, we consider how to modify a measure-uniform algorithm
    for general graphs to take advantage of 
the partition into
black and white components.
  Suppose we have a measure-uniform algorithm, $U$, for a distributed graph algorithm
  that can be divided into short phases. For example, the
    Greedy MIS Algorithm from Section~\ref{sec-css} can be
    divided into phases of length~$2$, each of which 
ends with an extendable partial solution.
  Then we can obtain another measure-uniform
  algorithm, $U_{bw}$, by alternately running phases on the black nodes and the white nodes.
  When $U$ is running on the black (white) nodes,
  it ignores the white (black) nodes, except that, before a black (white) node outputs 1 and terminates,
  it informs all its active neighbors, including its white (black) neighbors.
  If necessary, at the end of each phase,
  a clean-up algorithm is performed, in which every active node 
  that is neighbor of a node that output 1 informs its neighbors,
  outputs 0 and terminates.
  (Note that this is not necessary
  with the Greedy MIS Algorithm, since the clean-up is part of
  each phase.)
The round complexity of $U_{bw}$ is at most twice the round complexity of $U$.
However, it can take significantly fewer rounds when
the black and white components are %expected to be
much smaller than
 the error components found by the MIS Initialization Algorithm, 
 as in Figure~\ref{fig:grid}.
This measure-uniform algorithm could be combined with a reference
 algorithm, using whichever template is appropriate.

\subsection{An MIS Algorithm with Predictions for Rooted Trees}
\label{mis-tree}

In a rooted tree, each node knows whether it is the root and, if not, which of its neighbors is its parent.
For a strict binary rooted tree (in which all nodes are either leaves or have exactly two children),
it is possible to compute a maximal independent set in only four rounds~\cite{BBCOSST23}.
In this paper, we consider rooted trees where each node can have any number of children.

When restricted to graphs that are rooted trees, we use a different error measure,  $\eta_t$,
which is the maximum number of nodes on any monochromatic path from a node obtained by following parent pointers
in the subgraph induced by the active nodes at the end of the MIS Base Algorithm.
In other words, $\eta_t$ is 1 plus the maximum of the heights of the black and white components.
Note that this can be much smaller than the number of nodes in the largest
error component, and it is never larger.
Thus, for every rooted tree, $\eta_{t}\leq \eta_{bw} \leq \miseta$.

For rooted trees, one can use an initialization algorithm that ensures that active black nodes
are not adjacent to active white nodes when it is complete.
Then active black nodes and white nodes can be processed in parallel, without interfering
with one another. A measure-uniform algorithm, $U$, can directly take advantage
of the fact that black and white components may be smaller than the components induced
by the set of active nodes at the end of the initialization algorithm.

The MIS Rooted Tree Initialization algorithm
is a \good initialization algorithm for the MIS problem on rooted trees.
Its first two rounds are the same as in the MIS Base Algorithm, except
that the independent set $I$ chosen consists of the black nodes
(which received 1 as their prediction) 
that do not have a black parent. 
This is a superset of the set $I$ chosen by the  MIS Base Algorithm. In particular,
nodes can output~$1$, even if some of their children  are also black.
In round 3, each white node that received a message in round 2 sends 0 to its neighbors, outputs 0, and terminates.
Also, each white node that did not receive a message in round 2 and does not have a white parent sends 1 to its neighbors,
outputs 1, and terminates. In round 4, each active node that received a 1 in round 3, sends 0 to its neighbors, outputs 0, and terminates.
A node knows that the nodes it has  received messages from in rounds 2 or later are no longer active.
At the end of the MIS Rooted Tree \initalg{}, the  components of the subgraph induced by the active nodes are  
monochromatic.

An example for which this initialization algorithm is much better than the MIS Initialization Algorithm is
a directed line consisting of $3k$ nodes, where the white nodes are those at distance 0 modulo 3 from the 
root of the graph.
The independent set computed by the MIS Base Algorithm is empty, so $\miseta=3k$.
In contrast, the independent set $I$ computed by the MIS Rooted Tree Initialization
Algorithm consists of the $k$ nodes at distance 1 modulo 3 from the 
root, so all nodes terminate by the end of round 2.
Also note that, for this graph, $\eta_t = 2$.

\begin{algorithm}[ht]
\caption{A measure-uniform Algorithm for MIS in a Rooted Tree, code for node $p_i$}
\label{css-tree}
\begin{algorithmic}[0]
\For {$r = 1,\ldots$}
\State {\bf round} $2r-1$:
\If {$p_i$ is the root of its component (i.e., it does not have an active parent)}
\State it notifies all its active children that it is a root, 
\State outputs 1, and terminates
\ElsIf {$p_i$ is a leaf (i.e., it has no active children)} 
\State it notifies its parent that it is a leaf
\If {$p_i$ receives a message that its parent is a root}
\State it outputs 0 and terminates
\Else \ it outputs 1 and terminates
\EndIf
\EndIf
\State {\bf round} $2r$:
\If {$p_i$ received a message from a neighbor during round $2r-1$}
\State it notifies all its active neighbors that it is not in the independent set,
\State outputs 0, and terminates
\Else
\For {each neighbor $p_j$ from which $p_i$ received a message during round $2r$}
\State it locally records that $p_j$ has terminated
\EndFor
\EndIf
\EndFor
\end{algorithmic}
\end{algorithm}

Algorithm~\ref{css-tree} is a  simple measure-uniform algorithm that repeatedly adds
all roots and leaves to
the maximal independent set (except 
for a leaf whose parent is a root)
and then removes
these nodes and their neighbors from the graph.
The algorithm obtained from the Simple Template by first running the 
MIS Rooted Tree \initalg{} followed by Algorithm~\ref{css-tree} is consistent
and
has round complexity
at most  $\lceil \frac{\eta_t}{2} \rceil + 5$.

Goldberg, Plotkin and Shannon~\cite{GPS88} give a fault-tolerant algorithm for coloring a rooted tree 
with 3 colors in $O(\log^*d)$ rounds. 
Consider the reference algorithm, $R$, whose first part is their algorithm and whose
second part is a two round algorithm that outputs an independent set using this coloring.
In round 1, all nodes send their colors to their neighbors, all nodes with 
color 1
output 1 and terminate,
and all nodes that have a neighbor with 
color 1
output 0 and terminate.
In round 2,
all active nodes with
color 2
send a message to
their neighbors with
color 3,
output 1 and terminate.
An active node with 
 color 3
that receives a message in 
round 2
outputs 0 and terminates;  otherwise, it outputs 1 and terminates.

\begin{corollary}
\label{tree-cor}
  Using the \para, with the MIS Rooted Tree Initialization Algorithm as the
initialization algorithm, $B$,
 Algorithm~\ref{css-tree} as the measure-uniform algorithm, $U$,
 and the two part reference algorithm, $R$,  described above, gives
an MIS algorithm with predictions for rooted trees that
has consistency~3,
has round complexity
$\min\{\lceil \frac{\eta_t}{2} \rceil + 5, O(\log^* d)\}$
and is $\frac{\eta_t}{2}$-degrading.
When $\log^*d \in O(\log^*n)$,
this algorithm is robust and, hence, smooth.
\end{corollary}
\begin{proof}
  The MIS Rooted Tree Initialization Algorithm consists of 4 rounds, but it
  terminates in only 3 rounds if the predictions are correct.
We choose the upper bound, $r_1(d)$,  on the round complexity of the first part of $R$ to be even.
Then no clean-up is required, since the partial solution  output
at the end of every even numbered round of Algorithm~\ref{css-tree} is extendable.
Since the second part of $R$ has constant round complexity, by
Lemma~\ref{parallel}, the resulting algorithm with predictions is $\frac{\eta_t}{2}$-degrading and is robust with respect to $R$.
If the algorithm terminates during the first part, it performs at most $5+ \lceil \frac{\eta_t}{2} \rceil$ rounds.
Otherwise, $\lceil \frac{\eta_t}{2} \rceil > r_1(d)$ and it terminates in at most 2 additional rounds, for a total of at most $5+ r_1(d) + 2$ rounds.
Since   $r_1(d) \in O(\log^* d)$, the round complexity of this algorithm is $5 + \min\{\lceil \frac{\eta_t}{2} \rceil , O(\log^* d)\}$.

Any algorithm for solving the MIS problem on an $n$-node directed line (or ring) requires $\Omega(\log^* n)$ rounds~\cite{L92}.
If $\log^*d \in O(\log^*n)$, then $\Omega(\log^*d)$ rounds are required and $R$ is asymptotically optimal.
In this case, the resulting algorithm is robust and smooth.
\end{proof}

It is possible to obtain better algorithms with predictions for rooted trees than for general graphs.
For example, since $\eta_t \leq \miseta$,
the algorithm described in Corollary~\ref{tree-cor} is $\frac{\miseta}{2}$-degrading, whereas the  algorithm for general graphs described in Corollary~\ref{cor-parallel}  is $\miseta$-degrading. 
Because the round complexity of the reference algorithm used for rooted trees does not depend on $\Delta$, neither does
the round complexity of the resulting algorithm with predictions.

\section{Open Problems}
\label{section-open}

In general, it would be interesting to continue this work, considering other 
graph problems, other templates, other error measures, other classes of graphs, and possibly improving the
framework for distributed graph problems with predictions.
In particular, it would be interesting to further investigate 
the use of fault tolerant algorithms in the Parallel Template to design distributed graph algorithms with predictions.

A particularly interesting challenge would be to extend the work to
randomized algorithms.
The inherent difficulty is that if we use
an error measure based on the size of the largest remaining component,
we need to consider the expected maximum number of rounds over all components,
instead of the expected number of rounds in each (or the largest) component.

For example, consider Luby's randomized MIS algorithm~\cite{L86}, which has $O(\log n)$ expected round complexity.
If it is used as the reference algorithm in the Simple Template, the resulting algorithm with predictions has expected round complexity that is logarithmic in the sum
of the sizes of the error components, rather than logarithmic in $\miseta$, the size of the largest error component.
Although the expected round complexity for this reference algorithm to terminate any particular error component
is logarithmic in the size of that error component, the expected round complexity until the algorithm terminates on all error components can be considerably larger.
Specifically, suppose there are  $n/\log\log n$ error components, each of which is a path with $\log \log n$ nodes.
Of the $(\log \log n)!$ different possible orderings of the labels assigned to the nodes of a particular error component by Luby's algorithm,
there is one that is strictly increasing and one that is strictly decreasing, when 
going along the path.
In both these cases, the independent set chosen for this component consists of a single node and this node has a single neighbor.
So, with probability $2/(\log\log n)!$, a particular error component loses only one node in each of the first two rounds.
The probability that this happens in each subsequent iteration is at least as large, so the probability that the component loses exactly one node per round
is at least $(2/(\log\log n)!)^{(\log\log n)/2}$, which is much larger than $(\log \log n)/n$.
Thus, the expected number of such error components is at least one and the expected round complexity of the resulting algorithm with
predictions for this collection of error components is $\Theta(\log \log n)$, rather than $O(\log \miseta) = O(\log \log \log n)$.

In the area of online algorithms with predictions, there are many papers exhibiting trade-offs between consistency and robustness for various problems.
This includes
one of the very first papers in this area, by Kumar, Purohit, and Svitkina~\cite{KPS18} and the generalization by Lindermayr and Megow~\cite{LM22}. 
Algorithms of this type often have a parameter that can be adjusted to improve the consistency or improve the robustness, at the expense of the other.
Do such trade-offs exist for any problems in distributed graph algorithms with predictions? Are there interesting algorithms for some problems where the consistency is not determined using an initialization algorithm?

\section*{Acknowledgments}
We would like to thank Jukka Suomela for helpful discussion, including suggesting Lemma~\ref{css-vc} and its proof.

Joan Boyar and Kim S. Larsen were supported in part by the Independent Research Fund Denmark,
    Natural Sciences, grants DFF-0135-00018B and DFF-4283-00079B. Faith Ellen was supported in part by the Natural Science and Engineering Research Council of Canada grant RGPIN-2020-04178.

\bibliography{refs}

\begin{thebibliography}{10}

\bibitem{ALPS}
{Algorithms with Predictions}.
\newblock \url{https://algorithms-with-predictions.github.io/}, 2022.
\newblock Accessed: 2025-07-13.

\bibitem{ADDP19}
Karine Altisen, St{\'{e}}phane Devismes, Swan Dubois, and Franck Petit.
\newblock {\em Introduction to Distributed Self-Stabilizing Algorithms}.
\newblock Synthesis Lectures on Distributed Computing Theory. Morgan {\&}
  Claypool Publishers, 2019.

\bibitem{ABM24}
Antonios Antoniadis, Hajo Broersma, and Yang Meng.
\newblock Online graph coloring with predictions.
\newblock In {\em 8th International Symposium on Combinatorial Optimization
  ({ISCO}), Revised Selected Papers}, volume 14594 of {\em Lecture Notes in
  Computer Science}, pages 289--302. Springer, 2024.

\bibitem{APT22}
Yossi Azar, Debmalya Panigrahi, and Noam Touitou.
\newblock Online graph algorithms with predictions.
\newblock In {\em 33rd {ACM-SIAM} Symposium on Discrete Algorithms ({SODA})},
  pages 35--66. {SIAM}, 2022.

\bibitem{BBCOSST23}
Alkida Balliu, Sebastian Brandt, Yi{-}Jun Chang, Dennis Olivetti, Jan
  Studen{\'{y}}, Jukka Suomela, and Aleksandr Tereshchenko.
\newblock Locally checkable problems in rooted trees.
\newblock {\em Distributed Computing}, 36(3):277--311, 2023.

\bibitem{BBHORS21}
Alkida Balliu, Sebastian Brandt, Juho Hirvonen, Dennis Olivetti, Mika{\"{e}}l
  Rabie, and Jukka Suomela.
\newblock Lower bounds for maximal matchings and maximal independent sets.
\newblock {\em Journal of the {ACM}}, 68(5):39:1--39:30, 2021.

\bibitem{BBKNORS25}
Alkida Balliu, Sebastian Brandt, Fabian Kuhn, Krzysztof Nowicki, Dennis
  Olivetti, Eva Rotenberg, and Jukka Suomela.
\newblock Distributed computation with local advice.
\newblock In {\em 39th International Symposium on Distributed Computing
  ({DISC})}, pages 12:1--12:19, 2025.

\bibitem{BHMORS22}
Alkida Balliu, Juho Hirvonen, Darya Melnyk, Dennis Olivetti, Joel Rybicki, and
  Jukka Suomela.
\newblock Local mending.
\newblock In {\em 29th International Colloquium on Structural Information and
  Communication Complexity ({SIROCCO})}, volume 13298 of {\em Lecture Notes in
  Computer Science}, pages 1--20. Springer, 2022.

\bibitem{BE13}
Leonid Barenboim and Michael Elkin.
\newblock {\em Distributed Graph Coloring: Fundamentals and Recent
  Developments}.
\newblock Synthesis Lectures on Distributed Computing Theory. Morgan {\&}
  Claypool Publishers, 2013.

\bibitem{BEG22}
Leonid Barenboim, Michael Elkin, and Uri Goldenberg.
\newblock Locally-iterative distributed ({\(\Delta\)} + 1)-coloring and
  applications.
\newblock {\em Journal of the {ACM}}, 69(1):5:1--5:26, 2022.

\bibitem{BDEG25}
Naama Ben-David, Muhammad~Ayaz Dzulfikar, Faith Ellen, and Seth Gilbert.
\newblock Byzantine agreement with predictions.
\newblock In {\em 44th Annual {ACM} Symposium on Principles of Distributed
  Computing ({PODC})}, pages 3--14, 2025.

\bibitem{B25}
Magnus Berg.
\newblock Comparing the hardness of online minimization and maximization
  problems with predictions.
\newblock In {\em 19th International Joint Conference on Frontiers of
  Algorithmics ({IJTCS-FAW})}, volume 15828 of {\em Lecture Notes in Computer
  Science}, pages 33--48. Springer, 2025.

\bibitem{BBFL25}
Magnus Berg, Joan Boyar, Lene~M. Favrholdt, and Kim~S. Larsen.
\newblock Complexity classes for online problems with and without predictions.
\newblock In {\em 19th International Joint Conference on Frontiers of
  Algorithmics ({IJTCS-FAW})}, volume 15828 of {\em Lecture Notes in Computer
  Science}, pages 49--63. Springer, 2025.

\bibitem{BLMMSS22}
Giulia Bernardini, Alexander Lindermayr, Alberto Marchetti{-}Spaccamela, Nicole
  Megow, Leen Stougie, and Michelle Sweering.
\newblock A universal error measure for input predictions applied to online
  graph problems.
\newblock In {\em 35th Advances in Neural Information Processing Systems
  ({NIPS})}, 2022.

\bibitem{BKKKM17}
Hans{-}Joachim B{\"{o}}ckenhauer, Dennis Komm, Rastislav Kr{\'{a}}lovic,
  Richard Kr{\'{a}}lovic, and Tobias M{\"{o}}mke.
\newblock Online algorithms with advice: The tape model.
\newblock {\em Information and Computation}, 254:59--83, 2017.

\bibitem{BFKL23}
Joan Boyar, Lene~M. Favrholdt, Shahin Kamali, and Kim~S. Larsen.
\newblock Online interval scheduling with predictions.
\newblock In {\em 18th International Symposium on Algorithms and Data
  Structures ({WADS})}, volume 14079 of {\em Lecture Notes in Computer
  Science}, pages 193--207. Springer, 2023.

\bibitem{BFKLM17j}
Joan Boyar, Lene~M. Favrholdt, Christian Kudahl, Kim~S. Larsen, and Jesper~W.
  Mikkelsen.
\newblock {Online Algorithms with Advice: A Survey}.
\newblock {\em ACM Computing Surveys}, 50(2):1--34, 2017.
\newblock Article No.~19.

\bibitem{CSVZ22}
Justin~Y. Chen, Sandeep Silwal, Ali Vakilian, and Fred Zhang.
\newblock Faster fundamental graph algorithms via learned predictions.
\newblock In {\em 39th International Conference on Machine Learning ({ICML})},
  volume 162 of {\em Proceedings of Machine Learning Research}, pages
  3583--3602. {PMLR}, 2022.

\bibitem{DMVW23}
Sami Davies, Benjamin Moseley, Sergei Vassilvitskii, and Yuyan Wang.
\newblock Predictive flows for faster ford-fulkerson.
\newblock In {\em 40th International Conference on Machine Learning ({ICML})},
  volume 202 of {\em Proceedings of Machine Learning Research}, pages
  7231--7248. {PMLR}, 2023.

\bibitem{DVW24}
Sami Davies, Sergei Vassilvitskii, and Yuyan Wang.
\newblock Warm-starting push-relabel.
\newblock In {\em 37th Advances in Neural Information Processing Systems
  ({NIPS})}, 2024.

\bibitem{Dijkstra}
Edsger~W. Dijkstra.
\newblock Self-stabilizing systems in spite of distributed control.
\newblock {\em Communications of the {ACM}}, 17(11):643--644, 1974.

\bibitem{D09}
Stefan Dobrev, Rastislav Kr\'alovi\v{c}, and Dana Pardubsk{\'a}.
\newblock Measuring the problem-relevant information in input.
\newblock {\em RAIRO - Theoretical Informatics and Applications},
  43(3):585--613, 2009.

\bibitem{Dolev2000}
Shlomi Dolev.
\newblock {\em Self-Stabilization}.
\newblock {MIT} Press, 2000.

\bibitem{EG20}
Faith Ellen and Seth Gilbert.
\newblock Constant-length labelling schemes for faster deterministic radio
  broadcast.
\newblock In {\em 32nd ACM Symposium on Parallelism in Algorithms and
  Architectures {(SPAA)}}, pages 213--222, 2020.

\bibitem{EGMP21}
Faith Ellen, Barun Gorain, Avery Miller, and Andrzej Pelc.
\newblock Constant-length labeling schemes for deterministic radio broadcast.
\newblock {\em {ACM} Transactions on Parallel Computing}, 8(3):14:1--14:17,
  2021.

\bibitem{EFKR11}
Yuval Emek, Pierre Fraigniaud, Amos Korman, and Adi Ros{\'e}n.
\newblock Online computation with advice.
\newblock {\em Theoretical Computer Science}, 412(24):2642--2656, 2011.

\bibitem{FGIP09}
Pierre Fraigniaud, Cyril Gavoille, David Ilcinkas, and Andrzej Pelc.
\newblock Distributed computing with advice: Information sensitivity of graph
  coloring.
\newblock {\em Distributed Computing}, 21(6):395--403, 2009.

\bibitem{FIP10}
Pierre Fraigniaud, David Ilcinkas, and Andrzej Pelc.
\newblock Communication algorithms with advice.
\newblock {\em Journal of Computer and System Sciences}, 76(3-4):222--232,
  2010.

\bibitem{FKL10}
Pierre Fraigniaud, Amos Korman, and Emmanuelle Lebhar.
\newblock Local {MST} computation with short advice.
\newblock {\em Theory of Computing Systems}, 47:920--933, 2010.

\bibitem{GG24}
Mohsen Ghaffari and Christoph Grunau.
\newblock Near-optimal deterministic network decomposition and ruling set, and
  improved {MIS}.
\newblock In {\em 65th {IEEE} Annual Symposium on Foundations of Computer
  Science, ({FOCS})}, pages 2148--2179, 2024.

\bibitem{GGHIR23}
Mohsen Ghaffari, Christoph Grunau, Bernhard Haeupler, Saeed Ilchi, and
  V{\'{a}}clav Rozhon.
\newblock Improved distributed network decomposition, hitting sets, and
  spanners, via derandomization.
\newblock In {\em 34th Annual {ACM-SIAM} Symposium on Discrete Algorithms
  ({SODA})}, pages 2532--2566, 2023.

\bibitem{GNVW21}
Seth Gilbert, Calvin Newport, Nitin~H. Vaidya, and Alex Weaver.
\newblock Contention resolution with predictions.
\newblock In {\em 40th Annual {ACM} Symposium on Principles of Distributed
  Computing ({PODC})}, pages 127--137, 2021.

\bibitem{GPS88}
Andrew~V. Goldberg, Serge~A. Plotkin, and Gregory~E. Shannon.
\newblock Parallel symmetry-breaking in sparse graphs.
\newblock {\em {SIAM} Journal on Discrete Mathematics}, 1(4):434--446, 1988.

\bibitem{GoosS16}
Mika G{\"{o}}{\"{o}}s and Jukka Suomela.
\newblock Locally checkable proofs in distributed computing.
\newblock {\em Theory of Computing}, 12(1):1--33, 2016.

\bibitem{ramsey}
Ronald~L. Graham, Bruce~L. Rothschild, and Joel~H. Spencer.
\newblock {\em Ramsey Theory}.
\newblock John Wiley \& Sons, 1980.

\bibitem{HSSY24}
Monika Henzinger, Barna Saha, Martin~P. Seybold, and Christopher Ye.
\newblock On the complexity of algorithms with predictions for dynamic graph
  problems.
\newblock In {\em 15th Innovations in Theoretical Computer Science Conference
  ({ITCS})}, volume 287 of {\em LIPIcs}, pages 62:1--62:25. Schloss Dagstuhl -
  Leibniz-Zentrum f{\"{u}}r Informatik, 2024.

\bibitem{HKK10}
Juraj Hromkovi{\v{c}}, Rastislav Kr\'alovi\v{c}, and Richard Kr\'alovi\v{c}.
\newblock Information complexity of online problems.
\newblock In {\em 35th International Symposium on Mathematical Foundations of
  Computer Science (MFCS)}, volume 6281 of {\em Lecture Notes in Computer
  Science}, pages 24--36, Berlin, Heidelberg, 2010. Springer.

\bibitem{IKQP21}
Sungjin Im, Ravi Kumar, Mahshid~Montazer Qaem, and Manish Purohit.
\newblock Non-clairvoyant scheduling with predictions.
\newblock In {\em 33rd ACM Symposium on Parallelism in Algorithms and
  Architectures (SPAA)}, page 285–294. ACM, 2021.

\bibitem{IKQP21a}
Sungjin Im, Ravi Kumar, Mahshid~Montazer Qaem, and Manish Purohit.
\newblock Non-clairvoyant scheduling with predictions.
\newblock {\em {ACM} Transactions on Parallel Computing}, 10(4):19:1--19:26,
  2023.

\bibitem{KW85}
Richard~M. Karp and Avi Wigderson.
\newblock A fast parallel algorithm for the maximal independent set problem.
\newblock {\em Journal of the {ACM}}, 32(4):762--773, 1985.

\bibitem{KKP10}
Amos Korman, Shay Kutten, and David Peleg.
\newblock Proof labeling schemes.
\newblock {\em Distributed Computing}, 22(4):215--233, 2010.

\bibitem{KSV13}
Amos Korman, Jean{-}S{\'{e}}bastien Sereni, and Laurent Viennot.
\newblock Toward more localized local algorithms: removing assumptions
  concerning global knowledge.
\newblock {\em Distributed Computing}, 26(5-6):289--308, 2013.

\bibitem{KPS18}
Ravi Kumar, Manish Purohit, and Zoya Svitkina.
\newblock Improving online algorithms via {ML} predictions.
\newblock In {\em 31st Advances in Neural Information Processing Systems
  (NIPS)}, pages 9684--9693. Curran Associates, Inc., 2018.

\bibitem{KP99}
Shay Kutten and David Peleg.
\newblock Fault-local distributed mending.
\newblock {\em Journal of Algorithms}, 30(1):144--165, 1999.

\bibitem{KP00}
Shay Kutten and David Peleg.
\newblock Tight fault locality.
\newblock {\em {SIAM} Journal on Computing}, 30(1):247--268, 2000.

\bibitem{LM22}
Alexander Lindermayr and Nicole Megow.
\newblock Permutation predictions for non-clairvoyant scheduling.
\newblock In {\em 34th {ACM} Symposium on Parallelism in Algorithms and
  Architectures ({SPAA})}, pages 357--368. {ACM}, 2022.

\bibitem{L92}
Nathan Linial.
\newblock Locality in distributed graph algorithms.
\newblock {\em {SIAM} Journal on Computing}, 21(1):193--201, 1992.

\bibitem{L86}
Michael Luby.
\newblock A simple parallel algorithm for the maximal independent set problem.
\newblock {\em {SIAM} Journal on Computing}, 15(4):1036--1053, 1986.

\bibitem{LV18}
Thodoris Lykouris and Sergei Vassilvitskii.
\newblock Competitive caching with machine learned advice.
\newblock In {\em 35th International Conference on Machine Learning (ICML)},
  volume~80 of {\em Proceedings of Machine Learning Research}, pages
  3302--3311. {PMLR}, 2018.

\bibitem{LV21}
Thodoris Lykouris and Sergei Vassilvitskii.
\newblock Competitive caching with machine learned advice.
\newblock {\em Journal of the ACM}, 68(4):Article No.~24, 2021.

\bibitem{MMNNS25}
Samuel McCauley, Benjamin Moseley, Aidin Niaparast, Helia Niaparast, and Shikha
  Singh.
\newblock Incremental approximate single-source shortest paths with
  predictions.
\newblock In {\em 52nd International Colloquium on Automata, Languages, and
  Programming ({ICALP})}, volume 334 of {\em LIPIcs}, pages 117:1--117:20.
  Schloss Dagstuhl - Leibniz-Zentrum f{\"{u}}r Informatik, 2025.

\bibitem{MMNS24}
Samuel McCauley, Benjamin Moseley, Aidin Niaparast, and Shikha Singh.
\newblock Incremental topological ordering and cycle detection with
  predictions.
\newblock In {\em 41st International Conference on Machine Learning ({ICML})}.
  OpenReview.net, 2024.

\bibitem{MD22}
Michael Mitzenmacher and Matteo Dell'Amico.
\newblock The supermarket model with known and predicted service times.
\newblock {\em IEEE Transactions on Parallel and Distributed Systems},
  33(11):2740--2751, 2022.

\bibitem{MNS25}
Benjamin Moseley, Helia Niaparast, and Karan Singh.
\newblock Faster global minimum cut with predictions.
\newblock {\em CoRR}, abs/2503.05004, 2025.

\bibitem{Peleg2000}
David Peleg.
\newblock {\em Distributed Computing: A Locality-Sensitive Approach}.
\newblock SIAM Monographs on Discrete Mathematics and Applications. Society for
  Industrial and Applied Mathematics, 2000.

\bibitem{PZ24}
Adam Polak and Maksym Zub.
\newblock Learning-augmented maximum flow.
\newblock {\em Information Processing Letters}, 186:106487, 2024.

\bibitem{ramseya}
F.~P. Ramsey.
\newblock On a problem in formal logic.
\newblock {\em Proceedings of the London Mathematical Society}, 30:264--286,
  1930.

\bibitem{BFNP24}
Jan van~den Brand, Sebastian Forster, Yasamin Nazari, and Adam Polak.
\newblock On dynamic graph algorithms with predictions.
\newblock In {\em 35th {ACM-SIAM} Symposium on Discrete Algorithms ({SODA})},
  pages 3534--3557. {SIAM}, 2024.

\end{thebibliography}
\bibliographystyle{plain}

\end{document}